\documentclass[11pt]{article}

\usepackage{setspace}
\usepackage[english]{babel} 
\usepackage[utf8]{inputenc} 
\usepackage[a4paper,margin=1in]{geometry} 
\usepackage{amsmath} 
\usepackage{bm} 
\usepackage{graphicx} 
\usepackage[export]{adjustbox} 
\usepackage{pdflscape} 
\usepackage{fancyhdr} 
\usepackage[colorlinks=true,citecolor=blue,urlcolor=blue,linkcolor=black]{hyperref} 
\usepackage[numbers,sort&compress]{natbib}
\usepackage{flafter} 
\usepackage[framemethod=tikz]{mdframed} 
\usepackage{color} 
\usepackage{wrapfig} 
\usepackage{lipsum} 
\usepackage[capitalize]{cleveref}
\crefname{assumption}{Assumption}{Assumptions}
\usepackage{amsfonts}
\usepackage{braket}
\usepackage{quantikz}
\usepackage{booktabs}
\usepackage{multirow}
\usepackage{graphicx}
\usepackage{amssymb}
\usepackage{caption}
\usepackage{amsthm}
\usepackage{xcolor}
\usepackage{enumitem}

\newcommand{\be}{\begin{equation}}
\newcommand{\ee}{\end{equation}}
\newcommand{\bes}{\begin{equation*}}
\newcommand{\ees}{\end{equation*}}
\newcommand{\bea}{\begin{eqnarray}}
\newcommand{\eea}{\end{eqnarray}}
\newcommand{\beas}{\begin{eqnarray*}}
\newcommand{\eeas}{\end{eqnarray*}}
\newcommand{\bal}{\begin{aligned}}
\newcommand{\eal}{\end{aligned}}

\newcommand{\mb}{\mathbb}

\newcommand{\mc}{\mathcal}

\graphicspath{ {Images/} } 

\usepackage{multicol} 
\usepackage{float} 

\usepackage{aliascnt}

\newtheorem{theorem}{Theorem}[section]

\newaliascnt{lemma}{theorem}
\newtheorem{lemma}[lemma]{Lemma}
\aliascntresetthe{lemma}

\newaliascnt{corollary}{theorem}
\newtheorem{corollary}[corollary]{Corollary}
\aliascntresetthe{corollary}

\newaliascnt{proposition}{theorem}

\aliascntresetthe{proposition}

\newaliascnt{remark}{theorem}
\newtheorem{remark}[remark]{Remark}
\aliascntresetthe{remark}

\newaliascnt{definition}{theorem}
\newtheorem{definition}[definition]{Definition}
\aliascntresetthe{definition}

\newaliascnt{problem}{theorem}

\aliascntresetthe{problem}

\newaliascnt{assumption}{theorem}
\newtheorem{assumption}[assumption]{Assumption}
\aliascntresetthe{assumption}

\newaliascnt{result}{theorem}
\newtheorem{result}[result]{Result}
\aliascntresetthe{result}

\newcommand{\norm}[1]{\left\lVert #1\right\rVert}

\newcommand{\diag}{\operatorname{diag}}

\usepackage{xcolor}
\definecolor{ddkcolor}{HTML}{D926AC}

\title{Transducer-based linear combination of unitaries: \\theory and applications}
\author{%
    Dong An\thanks{Beijing International Center for Mathematical Research, Peking University. 
    \href{mailto:dongan@pku.edu.cn}{\texttt{dongan@pku.edu.cn}}.}
    \quad
    Dekuan Dong\thanks{Beijing International Center for Mathematical Research, Peking University. 
    \href{mailto:ddk@pku.edu.cn}{\texttt{ddk@pku.edu.cn}}.}
    \quad
    Changpeng Shao\thanks{Academy of Mathematics and Systems Science, Chinese Academy of Sciences. 
    \href{mailto:changpeng.shao@amss.ac.cn}{\texttt{changpeng.shao@amss.ac.cn}}.}
    \quad 
    Yuxin Zhang\thanks{Alfr\'ed R\'enyi Institute of Mathematics. 
    \href{mailto:zhang.yuxin@renyi.hu}{\texttt{zhang.yuxin@renyi.hu}}.}
    \quad
    Chenhao Zhao\thanks{School of Mathematical Sciences, Peking University. 
    \href{mailto:2601110059@stu.pku.edu.cn}{\texttt{2601110059@stu.pku.edu.cn}}.}%
}
\date{\today}

\begin{document}

\pagenumbering{arabic}
\maketitle

\begin{abstract}
Linear combination of unitaries (LCU) is a fundamental primitive in quantum algorithms, whose cost is typically governed by the most expensive unitary appearing in the combination. We develop a transducer-based LCU framework that reduces this worst-case dependence to a weighted average query complexity, when the constituent unitaries share access to a common set of primitive oracles.

Consider $A=\sum_j c_j U_j$, where $c_j>0$ and each unitary $U_j$ can be implemented using $C_j$ primitive queries. Given an upper bound $a\geq \|A\|$, our algorithm implements a block-encoding of $A/\alpha$ with rescaling factor $\alpha=\mathcal{O}(a)$, using $\widetilde{\mathcal{O}} (C_{\max}+\overline{C} {\lambda}/{a} )$ primitive queries, where $\lambda=\sum_j c_j$, $C_{\max}=\max_j C_j$, and $\overline{C}= {\sum_j c_j C_j}/{\lambda}$. By comparison, the standard LCU construction requires $\widetilde{\mathcal{O}} (C_{\max} {\lambda}/{a} )$ primitive queries. The improvement can therefore be substantial when costly unitaries have small weights and $\lambda/a$ is large. As applications, we obtain improved block-encodings of sparse matrices, leading to quantum algorithms for sparse Hamiltonian simulation and quantum linear systems with near-optimal query complexity up to polylogarithmic factors in all relevant parameters.

Our main technique is the transducer framework developed by Belovs, Jeffery, and Yolcu [Quantum, 8:1444 (2024)]. Here we develop a complementary operator-level theory tailored to block-encodings. We identify the resolvent norm $K$ as a key complexity measure for transducer implementation besides the existing catalyst complexity. We show that the transducer action can be converted into an $\epsilon$-approximate block-encoding using only $\mathcal{O}\left(K \log(1/\epsilon)\right)$ queries to the transducer, improving the precision dependence from polynomial to logarithmic.

\end{abstract}

\tableofcontents

\section{Introduction}

A fundamental challenge in quantum algorithm design is to harness unitary quantum dynamics to manipulate the broad class of non-unitary operators that arise across a wide range of scientific applications. Linear combination of unitaries (LCU) provides a simple and remarkably versatile solution to this problem by embedding such operators into larger unitary transformations, and has consequently become a basic primitive underlying a wide range of quantum algorithms, including quantum linear system solvers \cite{childs2017quantum}, Hamiltonian simulation \cite{berry2015simulating, berry2014exponential}, differential equations \cite{an2023linear,AnChildsLin2023,LowSomma2025}, quantum walks \cite{ApersChakrabortyNovoRoland2022, AmbainisGilyenJefferyKokainis2020}, 
ground state preparation \cite{KeenDumitrescuWang2021, GeTuraCirac2019, ChakrabortyHazraLiShaoWangZhang2025}, quantum eigenvalue transformations \cite{LowSu2024,AnChildsLinYing2024}, etc. 
It also serves as one of the basic building tools for constructing block-encodings used in quantum singular value transformation (QSVT) \cite{GilyenSuLowEtAl2019}, a general framework for polynomial transformations of matrices.

The basic LCU construction~\cite{ChildsWiebe2012} is conceptually simple. 
Suppose that $A=\sum_{j} c_j U_j$, where $c_j>0$ and each $U_j$ is an efficiently implementable unitary, and let
$\lambda=\sum_j c_j$. 
One first prepares an ancilla state whose amplitudes encode the square root of the coefficients, $\sqrt{c_j}$, and then applies a controlled selection operation (commonly called the select oracle) $\sum_j \ket{j} \bra{j} \otimes U_j$. 
Uncomputing the ancilla yields a larger unitary whose appropriate block is exactly $A/\lambda$. 
In this way, a generally non-unitary operator $A$ is embedded into a unitary quantum circuit and can subsequently be manipulated using standard quantum algorithmic primitives.

The efficiency of this construction, however, is determined not only by the normalization factor $\lambda$, but also by the cost of implementing the controlled selection operation. 
To be more exact, suppose each $U_j$ costs $C_j$ and under reasonable model assumptions, $\sum_j \ket{j} \bra{j} \otimes U_j$ can be implemented in cost $\mathcal{O}(C_{\max})$ with $C_{\max}=\max_j C_j$.\footnote{But generally, the cost is $\sum_j C_j$, here we assume it is dominated by the worst unitary for simplicity. This is also the common scenario in practice.}
Given a known upper bound $a\geq \|A\|$ with $0<a\leq\lambda$,
standard LCU followed by singular value amplification~\cite{GilyenSuLowEtAl2019} constructs an $\mathcal{O}(a)$-normalized block-encoding at a query cost of $\widetilde{\mathcal{O}}(C_{\max}\lambda/a)$.

One may realize that the coefficients of an LCU naturally define a probability distribution over its constituent operations. 
Indeed, writing $p_j=c_j/\lambda$, a classical randomized procedure that selects the $j$-th operator with probability $p_j$ would lead to the expectation $A/\lambda$, incurring the expected query cost $\overline{C}=\sum_j p_jC_j$. 
Thus LCU carries an intrinsic weighted-average notion of complexity. 
This average-cost structure is, however, largely lost in the standard coherent implementation of LCU. 
The select operation must implement all constituent operations coherently and is therefore typically governed by the worst-case cost $C_{\max}$. 
A natural question is whether one can retain quantum coherence while recovering the weighted-average cost inherent in the LCU decomposition.

In this work, we identify a general weighted-cost principle for LCU by developing a transducer-based LCU framework. 
Specifically, we reduce the standard LCU cost from $\widetilde{\mathcal{O}}(C_{\max} \lambda/a)$ to $\widetilde{\mathcal{O}}(C_{\max}+\overline{C}\lambda/a)$.
Therefore our approach can be advantageous when the high-cost unitaries have small coefficient weights and the coefficient $1$-norm is much larger than the spectral norm of the matrix $A$. 
As applications, we obtain an improved block-encoding of sparse matrices, which in turn yields quantum algorithms for sparse Hamiltonian simulation and sparse linear systems with query complexities matching the lower bounds of \cite{Low2019,mori2026sparsity,patel2026nearoptimaljointlowerbound} up to polylogarithmic factors. 

\subsection{Results}

\subsubsection{Main result}

Our main result gives a general LCU construction whose normalization can be reduced at a cost determined by the weighted average of the constituent query complexities. 

\begin{result}[Transducer-based LCU, informal version of~\cref{thm:complexity_main}]\label{res:LCU}
    Let $U_j$ be unitaries that can be implemented by $C_j$ queries to a fixed set of primitive oracles. 
    Let $A = \sum_{j=0}^{J-1} c_j U_j$ with coefficients $c_j > 0$, and suppose that an upper bound $a \geq \|A\|$ is known. 
    Then, there exists an algorithm that implements $A/\alpha$ with $\alpha = \mathcal{O}(a)$ up to error $\mathcal{O}(\epsilon)$, using 
    \begin{equation}
        \mathcal{O}\left( C_{\max} \log\left( \frac{1}{\epsilon} \right) + \frac{\lambda}{a} \overline{C} \left(\log\left(\frac{1}{\epsilon}\right)\right)^2 \right)
    \end{equation}
    queries to the primitive oracles, where $\lambda = \sum_{j=0}^{J-1} c_j ,\, C_{\max} = \max_j C_j, \,\overline{C} = (\sum_{j=0}^{J-1} c_j C_j)/\lambda$. 
\end{result}

\Cref{res:LCU} requires two conditions. 
First, the circuits for $U_j$'s are related in the sense that their oracle calls belong to a common family that can be selected coherently (called shared primitive access assumption, as specified in~\cref{assump:U}). 
While excluding the most general case when all the unitaries are unrelated, this condition can be satisfied in many practical scenarios, such as the quantum phase estimation circuit, and the Hamiltonian simulation and linear system algorithms that we will discuss in~\cref{sec:intro_applications}. 
Second, an upper bound of $\|A\|$ should be known a priori. 
Although we may simply take $a = \lambda$, such a choice cannot lead to a genuinely improved query complexity as being discussed in~\cref{sec:intro_comparison}. 

In fact, the main improvement concerns the cost of reducing the normalization from $\lambda$ to $\mathcal{O}(a)$. 
In the standard construction of LCU, it requires repeated uses of the entire select oracle, each costing the most expensive complexity $C_{\max}$. 
In our algorithm, the repeated cost is the averaged cost $\overline{C}$, while $C_{\max}$ appears only as an additive contribution. 
The saving can therefore be substantial when expensive unitaries carry small coefficient weights and $\lambda/a$ is large. 
The additive term also makes clear that our result does not imply that an arbitrary constituent circuit can be implemented at its average cost.

Our~\cref{res:LCU} applies more generally to $A = \sum_{j=0}^{J-1} c_j A_j$, where $A_j$ are general matrices supplied through their block-encodings. 
This general formulation allows our algorithm to be used as a subroutine in other algorithms based on block-encodings, such as the sparse matrix calculations considered in this paper.

\subsubsection{Applications}\label{sec:intro_applications}

We apply our transducer-based LCU to construct a block-encoding of a sparse matrix, which then serves as the common input model for sparse matrix functions such as Hamiltonian simulation and quantum linear system solving. 
We use the standard sparse-access model, with oracles that return the locations and values of nonzero entries. 
For the applications below, we consider Hermitian matrices with unit spectral norm. 

\begin{result}[Block-encoding of a sparse matrix, informal version of~\cref{cor:sparse_block_encoding_normalized}]\label{res:sparse-block-encoding}
    A block-encoding of a $d$-sparse normalized matrix can be constructed with $\mathcal{O}(1)$ rescaling factor and $\mathcal{O}(\sqrt{d}\, \operatorname{polylog}(d/\epsilon) )$ queries to its standard sparse access oracles. 
\end{result}

\Cref{res:sparse-block-encoding} gives a block-encoding whose rescaling factor remains constant while the number of sparse oracle queries grows only as $\sqrt{d}$, up to poly-logarithmic factors. 
A constant rescaling factor is useful because it avoids introducing an additional sparsity dependence into subsequent matrix transformations. 
The construction therefore provides a reusable representation of the matrix, with the query savings available to any algorithm that uses this block-encoding.

Our first application of this representation is Hamiltonian simulation. 
For a normalized Hamiltonian $A$, the goal is to implement $e^{-iAT}$ for a time $T$. 
Combining the sparse block-encoding with quantum singular value transformation (QSVT)~\cite{GilyenSuLowEtAl2019} yields nearly linear dependence on the evolution time $T$, with each use of the block-encoding costing $\widetilde{\mathcal{O}}(\sqrt{d})$ sparse oracle queries.

\begin{result}[Near-optimal sparse Hamiltonian simulation algorithm, informal version of~\cref{thm:sparse_Ham_sim}]\label{res:sparse_Ham_sim}
    There exists a quantum algorithm which simulates a $d$-sparse normalized Hamiltonian to time $T$ with error at most $\mathcal{O}(\epsilon)$, using $\mathcal{O}( \sqrt{d}\, T \operatorname{polylog}(dT/\epsilon))$ queries to the Hamiltonian sparse oracle. 
\end{result}

For a general Hamiltonian, a known spectral norm bound can be absorbed into the evolution time by rescaling. 
The leading $\sqrt{d}\,T$ dependence matches the worst-case lower bound of~\cite{Low2019} up to polylogarithmic factors. 
Thus \cref{res:sparse_Ham_sim} is near-optimal in its joint dependence on sparsity and time, and the accuracy dependence achieves a polylogarithmic scaling.  

The same block-encoding can also be combined with quantum linear system solvers. 
Given an invertible $d$-sparse Hermitian matrix $A$ with $\|A\|=1$, condition number $\kappa=\|A\|\|A^{-1}\|$, and a normalized right-hand-side state $\ket{b}$,
the goal is to prepare the normalized solution state $A^{-1}\ket{b}/\|A^{-1}\ket{b}\|$. 

\begin{result}[Near-optimal sparse linear system algorithm, informal version of~\cref{thm:sparse_linear_system}]\label{res:sparse_linear_system}
    There exists a quantum algorithm which solves the quantum sparse linear system problem $Ax = \ket{b}$ with error at most $\mathcal{O}(\epsilon)$, using $\mathcal{O}( \sqrt{d}\, \kappa \operatorname{polylog}(d\kappa/\epsilon))$ queries to the matrix sparse oracle and $\mathcal{O}(\kappa/\|A^{-1}\ket{b}\|)$ queries to the right-hand-side state. 
\end{result}

In~\cref{res:sparse_linear_system}, the matrix-query bound has leading dependence $\sqrt{d}\,\kappa \, \operatorname{polylog}(1/\epsilon)$, matching the joint lower bound of~\cite{patel2026nearoptimaljointlowerbound} in all parameters up to polylogarithmic factors. 
This dependence can follow from using our block-encoding in any existing quantum linear system solver whose matrix-query cost is nearly linear in $\kappa$ and poly-logarithmic in $1/\epsilon$. 
Here we use the tunable variable-time amplitude amplification (tunable VTAA) algorithm of~\cite{LowSu2024quantumlinearalgorithmoptimal}, which not only yields the desired matrix query complexity but also gives the $\mathcal{O}(\kappa/\|A^{-1}\ket{b}\|)$ scaling for the number of preparing $\ket{b}$, matching the lower bound for the state preparation cost also established in~\cite{LowSu2024quantumlinearalgorithmoptimal}. 
The ratio $\kappa/\|A^{-1}\ket{b}\|$ lies between $1$ and $\kappa$ and can be much smaller than $\kappa$, making this distinction useful when preparing $\ket{b}$ is expensive. 
In summary, our sparse block-encoding improves the matrix query complexity while preserving this favorable dependence on right-hand-side preparation.

In this work, we study only two representative applications of \Cref{res:sparse-block-encoding} in detail. More generally, when combined with QSVT, \Cref{res:sparse-block-encoding} leads to improved quantum algorithms for general functions of sparse matrices, yielding a near-quadratic improvement in the dependence on sparsity compared with the algorithm of \cite{MontanaroShao2024}. 
Moreover, by combining the low bound techniques of \cite{Low2019,MontanaroShao2024}, one can show that the square-root dependence on the sparsity is indeed optimal.

\subsection{Overview of the algorithm}

The guiding idea is to retain the intermediate stages of the constituent circuits throughout the LCU construction. 
To this end, a convenient way is through the transducer, which was proposed in~\cite{BelovsJefferyYolcu2024} and recently received rapidly increasing attention in quantum algorithm design~\cite{apers2026elfs,chen2026optimal,chen2026gateefficientimplementationqueryoptimaltimedependent,wang2026queryoptimalquantumsimulationlindblad,chen2026queryoptimalgateefficientlindbladiansimulation}. 
In a transducer, its state space has a public part, which carries the desired input and output, and a private part containing an auxiliary vector called the catalyst unchanged in the algorithm. 
It is an idealization of the usual bounded-error quantum algorithm that can exhibit useful properties such as efficient composition. 
Specifically, one can turn a circuit consisting of a sequential composition of oracles into a large unitary transducer with only a constant number of controlled oracle calls, similar to the select oracle in the standard LCU implementation. 
Nevertheless, the overall complexity, instead of being reduced to a constant level, is transferred to the norm of the catalyst, which may exhibit the average cost among the unitaries. 

Our LCU algorithm uses this representation in three steps. 
For clarity, we describe the Hermitian unitary case. 
The usual Hermitian dilation extends the construction to general matrices, and a complete technical overview of our algorithm is given in~\cref{sec:algorithm_description}. 

\begin{enumerate}[label=\textbf{Step \arabic*:}, leftmargin=*]
    \item \textbf{Combine the constituent unitaries as transducers.} 
    We first represent each unitary $U_j$, consisting of sequential oracle calls, by a transducer $S_j$ whose catalyst contains its intermediate states. 
    A single application of $S_j$ coherently calls the oracles in parallel (like the select oracle in standard LCU).  
    We then combine them coherently using coefficient amplitudes proportional to $\sqrt{c_j}$, so the squared size of the combined circuit histories is governed by $\overline{C}$. 
    A simple phase operation turns the construction into a transducer whose public action is the Cayley transformation of $A/(4\lambda)$, defined as $\operatorname{Cay}(A/(4\lambda))=(I-iA/(4\lambda))(I+iA/(4\lambda))^{-1}$. 
    Notice that here the rescaling factor of $A$ is still on the order of the 1-norm of the coefficients in LCU. 
    \item \textbf{Recover the target linear combination and improve its normalization.} 
    We next design a QSVT circuit that approximately inverts the Cayley transformation and amplifies the resulting block from rescaling factor $\mathcal{O}(\lambda)$ to $\mathcal{O}(a)$. 
    The length of such a QSVT circuit is $L = \mathcal{O}((\lambda/a)\log(1/\epsilon))$, the degree of the polynomial for approximate inverse Cayley transformation and amplification. 
    Crucially, we do not actually implement this QSVT circuit, but assemble a transducer for this entire circuit in a similar way as we perform in Step 1. 
    The public action of the resulting transducer is a block-encoding of $A/(16a)$, and the amplification ratio $\lambda/a$ enters the complexity through the catalyst norm, which is multiplied only on the average cost. 
    \item \textbf{Convert the transducer to the standard block-encoding.} 
    The catalyst is a mathematical description of the computation and need not be prepared. 
    In the standard procedure~\cite{BelovsJefferyYolcu2024}, one may instead start with zero amplitude in the private part, distribute the input coherently among labelled public sectors, and repeatedly apply the transducer while reusing the same private space. 
    Recombining the public sectors gives an approximate implementation of the desired action. 
    However, the basic construction converges only algebraically with the number of repetitions, leading to polynomial dependence on $1/\epsilon$.  
    To improve the accuracy dependence, we leverage the idea of extrapolation, which is a widely used methodology for constructing high-order numerical methods from lower-order ones in both classical numerical analysis and quantum algorithm design. 
    Specifically, we linearly combine multiple circuits from our algorithm with different transducer repeats at the final stage to cancel low-order error terms, leading to an approximation of the block-encoding of $A/(16a)$ with exponentially reduced error. 
\end{enumerate}

The complexity then follows from keeping track of two different contributions. 
The longest constituent circuit sets an initial scale $\mathcal{O}(C_{\max})$, while composing $L$ uses of the Cayley transformation in Step 2 adds $\mathcal{O}(L\overline{C})$. 
High-order extrapolation in Step 3 incurs one further factor of $\log(1/\epsilon)$. 
The resulting primitive query cost is therefore $\mathcal{O}\bigl((C_{\max}+L\overline{C})\log(1/\epsilon)\bigr)$, which gives the claimed scaling in~\cref{res:LCU}. 
At the technical level, the quantity controlling the final conversion is the resolvent norm of the private part of the transducer (see~\cref{thm:finite-reuse-high-order}). 
Its composition bound (see~\cref{cor:composition_general}) is what allows the longest circuit to contribute additively while the amplification cost accumulates with the weighted average.

\paragraph{Sparse matrix applications.}

For a sparse matrix $A$ with $\|A\| \leq 1$, motivated by~\cite{Low2019}, we start with a useful decomposition $A = \sum_{j=0}^m A_j$ that groups matrix entries according to their magnitudes, with adjacent thresholds differing by a factor of two. 
Then we block-encode each component $A_j$ by the standard sparse matrix block-encoding with amplitude amplification~\cite{Low2019}, and finally linearly combine them by our transducer-based LCU. 
This gives a block-encoding of the original sparse matrix $A$ with $\mathcal{O}(1)$ rescaling factor. 
See~\cref{sec:applications_sparse_block_encoding} for a detailed description of this approach. 

The overall sparsity dependence is only square root, up to polylogarithmic factors, for the following reason. 
In a sparse matrix, small entries can be numerous, whereas a bound on the Euclidean norm of each column limits how many large entries it can contain. 
This gives the tradeoff that groups containing larger entries require more queries but have smaller rescaling factors. 
These two effects compensate in the weighted sum of costs, which can be leveraged by our transducer-based LCU. 
More specifically, for a matrix of spectral norm at most one, this decomposition has $\lambda=\mathcal{O}(\sqrt{d})$, $C_{\max}=\mathcal{O}(\sqrt{d}\log d)$, and $\overline{C}=\mathcal{O}((\log d)^2)$ (See~\cref{lem:sparse_decomposition_norm}). 
Transducer-based LCU therefore gives a block-encoding with constant normalization using $\mathcal{O}(\sqrt{d}\,\operatorname{polylog}d)$ queries. 

With a block-encoding available, Hamiltonian simulation and linear system tasks can be readily solved by any existing quantum algorithms. 
We obtain our~\cref{res:sparse_Ham_sim} by combining transducer-based LCU with QSVT for Hamiltonian simulation~\cite{GilyenSuLowEtAl2019}, and~\cref{res:sparse_linear_system} by using the tunable VTAA~\cite{LowSu2024quantumlinearalgorithmoptimal}. 
In~\cref{sec:applications_sparse_linear_system}, we also discuss the complexity of combining the transducer-based LCU with the discrete adiabatic linear system solver~\cite{CostaAnSandersEtAl2022}, which is slightly better in the number of matrix queries but has worse state preparation cost.

\subsection{Related work and comparison}\label{sec:intro_comparison}

\paragraph{Standard LCU.}
The usual LCU construction prepares a coefficient state, applies the select oracle, and unprepares the coefficient state~\cite{ChildsWiebe2012,berry2015simulating}. 
In the block-encoding formulation, it yields normalization $\lambda$, which can be reduced using singular value amplification~\cite{GilyenSuLowEtAl2019}. 
Under the shared primitive access assumption, a select oracle can be implemented using $\mathcal{O}(C_{\max})$ queries. 
Repeating it for amplification gives a cost $\mathcal{O}(C_{\max}\lambda/a)$ for a constant error level. 
Our transducer-based LCU replaces this multiplicative dependence by $\mathcal{O}(C_{\max}+\overline{C}\lambda/a)$. 

Nevertheless, a more careful complexity comparison must take into account both the query complexity, which directly affects the cost per run, and the rescaling factor, which determines the success probability of necessary post-selection steps. 
In~\cref{tab:comparison_lcu}, we compare our transducer-based LCU algorithm to the standard LCU with or without uniform singular value amplification. 
In particular, we additionally present an ``overall query complexity'', i.e., the product of the query complexity and the rescaling factor, which can be a representative complexity estimate when applying LCU on an input state and extracting the output state with amplitude amplification. 
The comparison in~\cref{tab:comparison_lcu} hereby shows that our transducer-based LCU can be advantageous when expensive unitaries are less weighted and the norm of the combined matrix is much less than the coefficient $1$-norm.

\begin{table}[h]
    \renewcommand{\arraystretch}{2}
    \centering
    \scalebox{0.75}{
    \begin{tabular}{c|c|c|c}\hline\hline
        \textbf{Method} 
        & \textbf{Query complexity per run} 
        & \textbf{Rescaling}
        & \textbf{Overall query complexity} 
        \\\hline 

        Standard LCU~\cite{ChildsWiebe2012}
        & $\mathcal{O}(C_{\max})$
        & $\mathcal{O}(\lambda)$
        & $\mathcal{O}(\lambda C_{\max})$
        \\\hline
        Standard LCU with amplification~\cite{GilyenSuLowEtAl2019}
        & $\mathcal{O}\left( C_{\max} \frac{\lambda}{a} \log\left( \frac{1}{\epsilon} \right)  \right)$
        & $\mathcal{O}(a)$
        & $\mathcal{O}\left(\lambda C_{\max} \log\left(\frac{1}{\epsilon}\right)\right)$
        \\\hline
        Our work (\cref{thm:complexity_main})
        & $\mathcal{O}\left( C_{\max} \log\left( \frac{1}{\epsilon} \right) + \frac{\lambda}{a} \overline{C} \left(\log\left(\frac{1}{\epsilon}\right)\right)^2 \right)$
        & $\mathcal{O}(a)$
        & $\mathcal{O}\left( a C_{\max} \log\left( \frac{1}{\epsilon} \right) + \lambda \overline{C} \left(\log\left(\frac{1}{\epsilon}\right)\right)^2 \right)$
        \\\hline\hline
    \end{tabular}
    }
    \caption{Comparison with representative quantum algorithms for implementing $A = \sum_{j=0}^{J-1} c_j U_j$.  We assume $c_j > 0$, $U_j$ are constructed from a shared set of primitive oracles and oracle-independent gates, and each $U_j$ has query complexity $C_j$. The parameter $a \geq \|A\|$ is a known upper bound of $\|A\|$, $\lambda = \sum c_j$ is the coefficient $1$-norm (we assume $a \leq \lambda$), $C_{\max} = \max_j C_j$ denotes the maximum complexity, and $\overline{C} = \sum c_j C_j / \lambda$ denotes the weighted average cost. The ``query complexity per run'' is the number of oracle calls to construct a block-encoding of $A$, ``rescaling'' is the rescaling factor of this block-encoding, and ``overall query complexity'' refers to the product of the former two columns, which is a representative complexity estimate when applying LCU on an initial state with standard amplitude amplification. }
    \label{tab:comparison_lcu}
\end{table}

\paragraph{Hamiltonian simulation.}

Among existing Hamiltonian simulation algorithms, quantum signal processing, qubitization and QSVT achieve optimal query dependence $\mathcal{O}(\alpha T+\log(1/\epsilon))$ for a given block-encoding with rescaling factor $\alpha$~\cite{LowChuang2017,LowChuang2019,GilyenSuLowEtAl2019}. 
With the standard sparse block-encoding~\cite{GilyenSuLowEtAl2019}, the rescaling factor is $\mathcal{O}(d)$, and thus the query complexity to the sparse oracles is $\mathcal{O}(dT + \log(1/\epsilon))$, which scales linearly in the sparsity $d$. 
Compared to these approaches, our algorithm almost quadratically improves the sparsity dependence, at the expense of extra poly-logarithmic factors in $T$ and $1/\epsilon$.

For a sparse Hamiltonian simulation, ref.~\cite{Low2019} proposes an algorithm through a decomposition by entry magnitude and recursive interaction-picture simulation, obtaining $\mathcal{O}(\sqrt{d} \, T (\sqrt{d}\, T/\epsilon)^{o(1)})$ queries. 
Here the $o(1)$ power represents extra sub-polynomial but super-logarithmic factors. 
Our sparse Hamiltonian simulation builds on this decomposition, and further achieves a better complexity by reducing the extra $(\sqrt{d}\, T)^{o(1)}$ to $\operatorname{polylog}(dT)$ and improving the accuracy dependence from super-logarithmic to poly-logarithmic. 
We would like to remark that this comparison concerns the case when the matrix is normalized according to its spectral norm. 
In a more general case without such a normalization,~\cite{Low2019} appears to have a better norm dependence (i.e., the matrix $1$-to-$2$ norm) than ours (suggested by~\cref{thm:sparse_block_encoding_transducer}). 

We summarize this complexity comparison in~\cref{tab:comparison_Ham_sim}.

\begin{table}[h]
    \renewcommand{\arraystretch}{2}
    \centering
    \scalebox{0.75}{
    \begin{tabular}{c|c}\hline\hline
        \textbf{Method} 
        & \textbf{Query complexity} 
        \\\hline 
        QSP/Qubitization/QSVT~\cite{LowChuang2017,LowChuang2019,GilyenSuLowEtAl2019} & $\mathcal{O}(d \, T + \log(1/\epsilon)) $
        \\\hline
        Recursive interaction-picture simulation~\cite{Low2019} & $\mathcal{O}( (\sqrt{d} \, T)^{1+o(1)} / \epsilon^{o(1)} )$
        \\\hline
        Our work (\cref{thm:sparse_Ham_sim})& $\mathcal{O}(\sqrt{d}\, T  \operatorname{polylog}(dT/\epsilon) ) $
        \\\hline\hline
    \end{tabular}
    }
    \caption{Comparison with representative quantum algorithms for sparse Hamiltonian simulation $e^{-i A T}$. Here $A$ is a Hamiltonian with $\|A\| = 1$ and sparsity $d$, and we assume access to its standard sparse oracle. The parameter $\epsilon$ denotes the target operator norm error. }
    \label{tab:comparison_Ham_sim}
\end{table}

\paragraph{Quantum linear systems.}
Starting with the HHL algorithm~\cite{HarrowHassidimLloyd2009}, there have been substantial works on quantum linear system algorithms with improved condition number and precision dependence (see, e.g., \cite{Ambainis2012,childs2017quantum,chakraborty2019power,SubasiSommaOrsucci2019,CostaAnSandersEtAl2022,Dalzell2024}).
Among them, the discrete adiabatic algorithm~\cite{CostaAnSandersEtAl2022} and the kernel reflection approach~\cite{Dalzell2024} achieve $\mathcal{O}(\alpha\kappa\log(1/\epsilon))$ block-encoding queries for $\|A\|=1$ and normalization $\alpha$. 
Combining these methods with a standard sparse block-encoding gives a worst-case dependence linear in $d$. 
Our algorithm keeps the main track of the algorithm unchanged, but only alters the construction of the block-encoding input model. 
This again can almost quadratically improve the sparsity dependence at the expense of extra poly-logarithmic factors in $\kappa$ and $1/\epsilon$. 

Alternatively, the work~\cite{LowSu2024quantumlinearalgorithmoptimal} separately improves query complexity to the right-hand-side state using tunable VTAA. 
Combining their solver with our block-encoding preserves this benefit while reducing the number of queries to the sparse matrix access. 
Recent algorithms also refine the dependence on the spectral properties of a particular instance beyond its worst-case condition number~\cite{Li2025,DalzellLiSu2026}, and our results address the complementary dependence on sparsity. 

Tailored for sparse linear system solvers, the work~\cite{Low2019} also already obtained square-root sparsity dependence with subpolynomial overhead, namely $\mathcal{O}(\sqrt{d}\, \kappa(\sqrt{d}\, \kappa/\epsilon)^{o(1)})$. 
Our algorithm achieves $\mathcal{O}(\sqrt{d}\, \kappa \, \operatorname{polylog}(\kappa d/\epsilon))$, thus reducing the $(\sqrt{d}\, \kappa/\epsilon)^{o(1)}$ overhead to 
$\operatorname{polylog}(d\kappa/\epsilon)$. 

We summarize these comparisons in~\cref{tab:comparison_linear_system}. 

\begin{table}[h]
    \renewcommand{\arraystretch}{2}
    \centering
    \scalebox{0.75}{
    \begin{tabular}{c|c|c}\hline\hline
        \textbf{Method} 
        & \textbf{Queries to the matrix oracles} 
        & \textbf{Queries to the state preparation}
        \\\hline 
        Discrete adiabatic~\cite{CostaAnSandersEtAl2022}/ Kernel reflection~\cite{Dalzell2024}
        & $\mathcal{O}(d\,\kappa \log(1/\epsilon))$
        & $\mathcal{O}(d\,\kappa \log(1/\epsilon))$
        \\\hline
        Tunable VTAA~\cite{LowSu2024quantumlinearalgorithmoptimal}
        & $\mathcal{O}(d\,\kappa \log \kappa \log(\log \kappa/\epsilon))$
        & $\mathcal{O}(\kappa/\|A^{-1}\ket{b}\|)$
        \\\hline
        Recursive interaction-picture simulation~\cite{Low2019}
        & $\mathcal{O}(  (\sqrt{d} \, \kappa)^{1+o(1)}/\epsilon^{o(1)} )$
        & $\mathcal{O}(\kappa/\|A^{-1}\ket{b}\|)$
        \\\hline
        Our work (\cref{thm:sparse_linear_system}) 
        & $\mathcal{O}( \sqrt{d}\, \kappa \operatorname{polylog}(d\kappa/\epsilon) )$
        & $\mathcal{O}(\kappa/\|A^{-1}\ket{b}\|)$
        \\\hline\hline
    \end{tabular}
    }
    \caption{Comparison with representative quantum algorithms for solving linear system $Ax = b$. Here $A$ is a Hermitian matrix with $\|A\| = 1$, sparsity $d$, and condition number $\kappa$. The right hand side $b$ is a normalized vector. 
    We assume access to the standard sparse oracle of $A$, and the state preparation oracle of $\ket{b}$. 
    The numbers of queries to these different oracles are given separately in the latter two columns. 
    The parameter $\epsilon$ denotes the target error in the normalized state measured by the vector $2$-norm. }
    \label{tab:comparison_linear_system}
\end{table}

\paragraph{Transducers.}
The main technical tool of our algorithm is the transducer framework introduced in~\cite{BelovsJefferyYolcu2024}, as a framework for quantum state conversion and composition, with costs that reflect how subroutines are used. 
Very recently, several transducer-based quantum algorithms have been proposed, including time-dependent Hamiltonian simulation~\cite{chen2026optimal,chen2026gateefficientimplementationqueryoptimaltimedependent}, Lindblad simulation~\cite{wang2026queryoptimalquantumsimulationlindblad}, quantum walks~\cite{apers2026elfs} and
ground state energy estimation~\cite{jeffery2026optimal}. 
Our work further extends the applicable regime of the transducer framework to the fundamental LCU task. 
We also establish two technical tools in the transducer framework with improved performance, which might be of independent interest. 

First, we develop a general high-order method that converts a transducer back into a standard block-encoding of its public action.
The general implementation in~\cite{BelovsJefferyYolcu2024} uses $\mathcal{O}(W/\epsilon^2)$ calls to the transducer to achieve state error $\epsilon$, where $W$ bounds the squared catalyst norm (which can be roughly understood as the weighted cost). 
In contrast, if the private block $T$ satisfies that
$I-T$ is invertible, and a bound
$K\geq\|(I-T)^{-1}\|$ on the resolvent is known, our construction establishes an $\mathcal{O}(K\log(1/\epsilon))$ complexity.
This achieves an exponential improvement in precision, but the norm dependence is changed from $W$ to $K$, the norm of the resolvent and might be worse than the catalyst norm in the worst case. 
For our LCU construction, however, the resolvent bound retains the weighted cost and enables logarithmic precision dependence.
Notice that such an exponential improvement has been previously achieved in~\cite{chen2026optimal}, while relying more delicately on the specific structure of the matrices. 

To bound the resolvent norm $K$ throughout the algorithm, we develop a more detailed analysis on the transducer of unitary and/or transducer composition in~\cref{sec:transducer_toolkit_composition}. 
This extends the result in~\cite{BelovsJefferyYolcu2024} with a complementary operator-level theory tailored to block-encodings and a rigorous norm bound on the resolvent. 

Besides technical tools, we also describe a simple Schur-complement interpretation of transducers. 
Roughly speaking, a transducer embeds the desired transformation into a larger unitary system, from which the effective operation is recovered by eliminating auxiliary variables. 
This viewpoint connects the transducer framework to standard Gaussian block elimination in linear algebra and provides useful intuition for both its composition rules and our complexity analysis.

\paragraph{Concurrent work.}

Bravo-Prieto, Harrow, and Kothari~\cite{bravoprieto2026optimalquantumlinearsystems} obtain the optimal sparse-matrix query complexity
$\Theta(\kappa\sqrt d\log(1/\epsilon))$ for quantum linear systems problem.
They also construct a block-encoding of
any $d$-sparse Hermitian matrix with spectral norm at most one, using
$\mathcal{O}(\sqrt d\log(1/\epsilon))$ queries.
Their method directly block-encodes an enlarged linear system whose inverse contains the desired operator or solution, further removing extra poly-logarithmic factors compared to our~\cref{res:sparse-block-encoding} and~\cref{res:sparse_linear_system}. 
Our main result addresses a different task of linear combinations of matrices, which applies beyond sparse-matrix access and provides a weighted-cost LCU primitive, together with resolvent bounds and high-precision transducer conversion.

\subsection{Organization}
The remainder of this paper is organized as follows. \Cref{sec:preliminaries} introduces the problem setting and reviews the basic tools of quantum linear algebra and transducers, together with a Schur-complement interpretation of the transducer framework. \Cref{sec:transducer_toolkit} develops the transducer toolkit used throughout the paper, including composition rules with resolvent-norm bounds and a high-order finite-reuse procedure for converting a transducer into a block-encoding of its action. \Cref{sec:transducer_LCU} presents the transducer-based LCU algorithm and its complexity estimate. \Cref{sec:applications} applies the framework to sparse matrices, deriving an improved sparse block-encoding and its applications to sparse Hamiltonian simulation and quantum linear systems. \Cref{sec:conclusion} concludes the work with potential future directions.

\section{Preliminaries}\label{sec:preliminaries}

In this section, we define the goal and conditions of the linear combination of matrices task considered in this work. 
We also review some basic concepts and results in quantum linear algebra and transducers, which will be the building blocks of our algorithm.

\subsection{Problem setting}

Consider a linear combination of matrices in the form of 
\begin{equation}\label{eqn:def_lcu}
    A = \sum_{j=0}^{J-1} c_j A_j. 
\end{equation}
Here $A_j$'s are general matrices, and $c_j$'s are positive real coefficients\footnote{The positive real assumption does not lose generality, as the phase factor of any potential complex coefficients can be absorbed in the matrix $A_j$. }. 
The main task of this work is to construct a block-encoding of $A$. 

We assume access to the block-encoding $U_j$ of $A_j/\alpha_j$ with normalization factor $\alpha_j \geq \|A_j\|$. 
Additionally, we make an important assumption that $U_j$ are related in the following sense. 

\begin{assumption}[Shared primitive access]\label{assump:U}
    The circuits of $U_j$ have known descriptions with $C_j$ calls to a fixed finite family of primitive oracles. 
    Specifically, each $U_j$ has decomposition $U_j =   G_j^{[C_j]} Q_j^{[C_j]} \cdots G_j^{[1]} Q_j^{[1]} G_j^{[0]}$, where $G_j^{[k]}$'s are known oracle-free unitaries, $Q_j^{[k]}$ belong to a fixed finite primitive oracle family and its adjoints, and assume that the coherent operator $\sum_{j,k} \ket{j,k}\bra{j,k}\otimes Q_j^{[k]}$ can be implemented with $\mathcal{O}(1)$ primitive queries. 
\end{assumption}

Throughout this work, the query complexity is counted as the number of queries to the primitive oracles in~\cref{assump:U}. 
This assumption allows the circuits $U_j$ to have different lengths and oracle-free gates while sharing a fixed family of primitive oracles.
It requires their next primitive call to be selected coherently from the labels using $\mathcal{O}(1)$ queries.

Assumption~\ref{assump:U} is natural in many structured settings. 
For example, the quantum phase estimation circuit consists of the controlled $U^j$ applications for the unitary $U$, and the circuits $U_j = U^j$ satisfy our shared primitive oracle assumption. 
Another example is multiple QSVT circuits based on a common block-encoding $U_H$, which may have different polynomial degrees and phase sequences, but their oracle calls use only $U_H$ and $U_H^\dagger$~\cite{GilyenSuLowEtAl2019}.
With controlled access to these primitives, they satisfy the assumption. 

This setting also includes the main applications considered in this work. In both sparse Hamiltonian simulation and sparse linear-system solving, the relevant block-encodings are ultimately constructed from the same sparse-access oracles of the input matrix. Thus the family of circuits appearing in these applications naturally shares a common primitive-query interface and satisfies Assumption~\ref{assump:U}. The assumption therefore does not require the \(U_j\)'s themselves to have the same circuit structure; it only requires their oracle-dependent components to arise from a common coherently accessible set of primitive oracles.

For the coefficients, we assume that they are known in the sense that preparing a quantum state encoding the information of $\sqrt{c_j}$ is efficient. For arbitrary unstructured coefficients, such a state can be prepared using $\mathcal O(J)$ elementary gates by standard generic state-preparation procedures \cite{mottonen2004transformationquantumstatesusing}, while substantially more efficient implementations are possible when the coefficients possess additional structure \cite{GroverRudolph2002}.
Nevertheless, they are independent of the primitive oracles for $U_j$, so the usage of the coefficients will not be counted in the query complexity.

In the remaining part of this work, we will only consider the special case where both $A_j$ and $U_j$ are Hermitian. 
This additional assumption does not sacrifice the generality due to the Hermitianization technique~\cite{chakraborty2019power}. 
Specifically, append one dilation qubit and let 
\begin{equation}
    \widetilde{A}_j = \ket{0}\bra{1}\otimes A_j + \ket{1}\bra{0}\otimes A_j^{\dagger} = \left( \begin{array}{cc}
        0 & A_j \\
        A_j^{\dagger} & 0
    \end{array} \right)
\end{equation}
and 
\begin{equation}
    \widetilde{U}_j = \ket{0}\bra{1}\otimes U_j + \ket{1}\bra{0}\otimes U_j^{\dagger}. 
\end{equation}
Then $\widetilde{A}_j$ is Hermitian, and $\widetilde{U}_j$ is a Hermitian unitary.
Furthermore, we have $(I\otimes \bra{0} \otimes I) \widetilde{U}_j (I\otimes \ket{0} \otimes I) =  \widetilde{A}_j/\alpha_j$, i.e., $\widetilde{U}_j$ is a Hermitian block-encoding of $\widetilde{A}_j/\alpha_j$ after reordering the sub-blocks. 
Therefore, once we know how to implement the linear combination of Hermitian matrices via their Hermitian block-encodings, we can directly obtain a block-encoding $\widetilde{U}$ of $\widetilde{A}/\widetilde{\alpha}$ with 
\begin{equation}
    \widetilde{A} = \sum_{j=0}^{J-1} c_j \widetilde{A}_j = \left( \begin{array}{cc}
        0 & A \\
        A^{\dagger} & 0
    \end{array} \right). 
\end{equation}
Then $\widetilde{U} (X\otimes I)$ (and regarding the dilation qubit as the ancilla qubit as well) gives the desired block-encoding of $A$.

\subsection{Quantum linear algebra primitives}

Block-encodings provide a standard way to represent a possibly non-unitary operator as a block of a larger unitary, enabling one to apply quantum algorithms, such as QSP and QSVT, to general linear operators. The notion was implicit in earlier Hamiltonian simulation algorithms~\cite{LowChuang2019} and was later formalized in~\cite{chakraborty2019power,GilyenSuLowEtAl2019}. We use the following standard definition.
\begin{definition}[Block-encoding]\label{def:block-encoding}
   Let $A$ be an $s$-qubit operator, $\alpha,\epsilon\in\mb R_{\geq 0}$ and $a\in \mb N$. An $(s+a)$-qubit unitary $U_A$ is called an $(\alpha, a, \epsilon)$-block-encoding of $A$, if 
   \be
   \left\|A-\alpha\cdot(\bra{0^a} \otimes I) U_A (\ket{0^a} \otimes I) \right\| \leq \epsilon.
   \ee
   Intuitively,
    \be
    U_A = \begin{pmatrix}
        A/\alpha & * \\ * & *
    \end{pmatrix}.
    \ee
\end{definition}

For a unitary decomposition $A=\sum_{j=0}^{J-1} c_j U_j$, the standard LCU~\cite{ChildsWiebe2012,berry2015simulating} framework yields a block-encoding of $A$ as follows.
\begin{lemma}[Standard LCU]\label{lem:standard_lcu}
    Let $A=\sum_{j=0}^{J-1} c_j U_j$, where each $U_j$ is unitary, $c_j>0$ and $c = \sum_j c_j$. Let $V$ be a unitary such that
    \be
    V \ket{0} = \frac{1}{\sqrt{c}} \sum_{j=0}^{J-1} \sqrt{c_j} \, \ket{j}.
    \ee
    Then the quantum circuit 
    \be
    (V^\dag \otimes I)\Big(\sum_{j=0}^{J-1} \ket{j}\bra{j} \otimes U_j \Big) (V\otimes I)
    \ee
    is an $(c, \lceil \log J \rceil, 0)$-block-encoding of $A$. It uses one application of $V, V^\dagger$, and one controlled application of each $U_j$.
\end{lemma}

Given a block-encoding of a Hermitian matrix, QSVT allows one to implement polynomial transformations of it.
\begin{lemma}[QSVT, Theorem 56 of \cite{GilyenSuLowEtAl2019}]\label{lem:QSVT}
    Suppose that $U$ is an $(\alpha, a, \epsilon)$-block-encoding of a Hermitian matrix $A$. Let $P \in \mathbb{R}[x]$ be a degree-$d$ polynomial satisfying $\left|P(x)\right| \leq {1}/{2}$ for all $x \in[-1,1]$.
    Then there is a quantum circuit $\widetilde{U}$, which is a $(1, a+2,4 d \sqrt{\epsilon / \alpha})$-block-encoding of $P(A / \alpha)$, and consists of $d$ applications of $U$ and $U^{\dagger}$, a single application of controlled-$U$ and $\mathcal{O}((a+1) d)$ other one- and two-qubit gates.
\end{lemma}

\subsection{Transducer}

A transducer represents a unitary transformation through a unitary $S$ acting on a direct sum of public and private spaces.
It transforms the public input as desired while returning an input-dependent private vector, called the catalyst, unchanged.
Here we give a formal definition of the transducer framework
of~\cite{BelovsJefferyYolcu2024}. 

\begin{definition}[Transducer]\label{def:transducer}
    Let $\mathcal{H}_{\rm{pub}}$ and $\mathcal{H}_{\rm{priv}}$ be two finite-dimensional Hilbert spaces. 
    A unitary $S: \mathcal{H}_{\rm{pub}}\oplus \mathcal{H}_{\rm{priv}} \rightarrow \mathcal{H}_{\rm{pub}}\oplus \mathcal{H}_{\rm{priv}}$ is called a \emph{transducer} of $V: \mathcal{H}_{\rm{pub}} \rightarrow \mathcal{H}_{\rm{pub}}$, if there exists a linear mapping $\Gamma: \mathcal{H}_{\rm{pub}} \rightarrow \mathcal{H}_{\rm{priv}}$ such that 
    \begin{equation}
        S(\psi \oplus \Gamma \psi ) = V \psi \oplus \Gamma \psi, \quad \forall \psi \in  \mathcal{H}_{\rm{pub}}. 
    \end{equation}
\end{definition}

Here $\mathcal H_{\mathrm{pub}}$ and $\mathcal H_{\mathrm{priv}}$
are called the public and private spaces, respectively. 
We call $V$ the \emph{public action}, $\Gamma\psi$ the
\emph{catalyst}, and $\Gamma$ the \emph{catalyst map}.
The catalyst vector is returned unchanged by $S$. For a fixed choice of catalyst map $\Gamma$, we define the input-dependent transduction complexity and its worst-case value by
\be
    W(S,\psi):=\|\Gamma\psi\|^2,
    \qquad
    w(S):=\sup_{\|\psi\|=1}W(S,\psi)=\|\Gamma\|^2.
\ee
When the catalyst map is not unique, these quantities refer to
the specified choice of $\Gamma$.
The catalyst serves as a mathematical witness and need not be prepared. The following result approximately implements the public
action using repeated applications of the transducer $S$, without access to $\Gamma$.

\begin{lemma}[Transducer implementation,
Theorem~5.5 of~\cite{BelovsJefferyYolcu2024}]
\label{lem:trans-to-alg}
    Let $S$ be a transducer with public action $V$ and catalyst
    map $\Gamma$. For every integer $N\geq1$ and initial state $\ket{\psi} \in \mc H_{\rm{pub}}$, there is a quantum
    circuit that outputs $\ket{\tau}$ such that
    \be
        \left\|
        \ket{\tau}-V\ket{\psi}
        \right\|
        \leq
        2\sqrt{\frac{W(S,\psi)}{N}}.
    \ee
    The circuit uses $N$ controlled applications of $S$
    and $\mathcal O(N)$ other elementary operations.
\end{lemma}

From the above result, implementing $V$ to additive error $\epsilon$ requires $\mathcal{O}(W/\epsilon^2)$ calls to $S$.
Later, in \cref{thm:finite-reuse-high-order}, we give a different implementation whose precision dependence is only logarithmic when the private-block resolvent is suitably bounded.

Some basic properties of the transducer directly follow from its definition. 
\begin{lemma}[Basic properties of a transducer]\label{lem:basic_prop_transducer}
    Let $S$ be a transducer of $V$ with catalyst map $\Gamma$, and its matrix representation under the standard basis is 
    \begin{equation}
        S = \left( \begin{array}{cc}
            D & E \\
            B & T
        \end{array} \right). 
    \end{equation}
    Suppose that $I-T$ is invertible, and define 
    \begin{equation}
        K(S) = \|(I-T)^{-1}\|. 
    \end{equation}
    Then: 
    \begin{enumerate}
        \item The action $V$ is unitary, and has the representation $V = D + E (I-T)^{-1} B$. 
        \item The catalyst has the representation $\Gamma = (I-T)^{-1}B$
        \item We have the following identities: $E(I-T)^{-1} = V \Gamma^{\dagger}$, and  $\Gamma\Gamma^{\dagger} = (I-T)^{-1} + (I-T^{\dagger})^{-1} - I$.  
        \item $w(S) \leq 2 K(S) - 1$. 
    \end{enumerate}
\end{lemma}
\begin{proof}The proof follows from elementary block-matrix manipulations together with the unitarity of $S$.
By definition, for every $\psi\in\mathcal H_{\rm pub}$ we have
\be
S
\begin{pmatrix}
\psi\\
\Gamma\psi
\end{pmatrix}
=
\begin{pmatrix}
V\psi\\
\Gamma\psi
\end{pmatrix}.
\ee
Since $S$ is unitary,
$
\|\psi\|^2+\|\Gamma\psi\|^2
=
\|V\psi\|^2+\|\Gamma\psi\|^2,
$
and hence
$
\|V\psi\|=\|\psi\|
$. Thus $V$ is unitary.
Writing the defining relation blockwise gives
\be
D+E\Gamma=V,
\qquad
B+T\Gamma=\Gamma.
\ee
The second identity gives
$
\Gamma=(I-T)^{-1}B.
$
Substituting this into the first identity yields
$
V
=
D+E(I-T)^{-1}B.
$
Since $S$ is unitary,
\be
S^\dagger
\begin{pmatrix}
V\psi\\
\Gamma\psi
\end{pmatrix}
=
\begin{pmatrix}
\psi\\
\Gamma\psi
\end{pmatrix}.
\ee
The second block equation gives
$
E^\dagger V+T^\dagger\Gamma=\Gamma,
$
and so
$
E(I-T)^{-1}=V\Gamma^\dagger.
$

Let
$
R:=(I-T)^{-1}.
$
Since $\Gamma=RB$ and $S$ is unitary, the lower-right block of
$SS^\dagger=I$ gives
$
BB^\dagger+TT^\dagger=I.
$
Hence
$
\Gamma\Gamma^\dagger
=
RBB^\dagger R^\dagger
=
R(I-TT^\dagger)R^\dagger.
$
Using
$
I-TT^\dagger
=
(I-T)+(I-T^\dagger)
-(I-T)(I-T^\dagger),
$
we obtain
$
\Gamma\Gamma^\dagger
=
(I-T)^{-1}
+
(I-T^\dagger)^{-1}
-I.
$
Finally, since $\Gamma\Gamma^\dagger$ is positive semidefinite,
\be
w(S)
=
\max_{\|x\|=1}
\langle x,\Gamma\Gamma^\dagger x\rangle
=
\max_{\|x\|=1}
\left(
2\operatorname{Re}\langle x,Rx\rangle-1
\right)
\leq
2\|R\|-1
=
2K(S)-1.
\ee
This proves the claimed bound.
\end{proof}

\subsection{Schur-complement interpretation of transducers}
\label{subsec:schur-transducer}

The algebraic structure underlying a transducer is essentially a Schur-complement reduction.  Schur complements form a classical and extensively studied subject in
matrix analysis, with applications ranging from block Gaussian elimination and
matrix inversion to numerical linear algebra, optimization, probability, control
theory, and effective operator constructions \cite{zhang2005schur}.  From this viewpoint, the transducer action can be understood through the effective operator obtained by eliminating the private variables.  This perspective provides a unified interpretation of many of the
basic identities and composition rules used below.

Let
\be
S=
\begin{pmatrix}
D&E\\
B&T
\end{pmatrix}
\ee
be a transducer and assume that $I-T$ is invertible.
The defining relation
\be
S
\begin{pmatrix}
x\\
\Gamma x
\end{pmatrix}
=
\begin{pmatrix}
Vx\\
\Gamma x
\end{pmatrix}
\ee
admits a natural interpretation in terms of block elimination.
Indeed,
\begin{equation}
\label{a key formula of Schur}
(I-S) \begin{pmatrix}
    x \\
    \Gamma x
\end{pmatrix}
= \begin{pmatrix}
    (I-V) x \\
    0
\end{pmatrix}.
\end{equation}
This means that eliminating the private component amounts to taking the Schur complement of the lower-right block $I-T$ in $I-S$. Explicitly,
\be
\operatorname{Schur}_{I-T}(I-S)
=
(I-D)-(-E)(I-T)^{-1}(-B)
=
I-V.
\ee
In other words, $I-V$ is precisely the reduced operator obtained by eliminating the private variables from $I-S$, while $V$ is the corresponding action of the transducer on the public space.

Concretely, for a given public vector $x$, the second block equation uniquely determines the private vector that must accompany it,
$
y=(I-T)^{-1}Bx=\Gamma x.
$
Substituting this solution into the first block equation then yields the reduced action on the public space. This is precisely the usual interpretation of the Schur complement: one first solves for the variables to be eliminated and then substitutes them back to obtain the effective equation for the variables that are retained.

The preceding relation can equivalently be written as the block elimination identity
\be
\begin{pmatrix}
I&E(I-T)^{-1}\\
0&I
\end{pmatrix}
(I-S)
=
\begin{pmatrix}
I-V&0\\
-B&I-T
\end{pmatrix}.
\ee
Thus, by eliminating the private variables, $I-S$ is reduced in a single block Gaussian elimination step to a block-triangular matrix whose diagonal blocks are $I-V$ and $I-T$. In this sense, the reduced operator $I-V$, and hence the action $V$ of the transducer on the public space, arises directly from block Gaussian elimination of the private variables.

This also clarifies the roles of $K(S)$ and the catalyst map $\Gamma$. The matrix $I-T$ serves as the pivot in the block elimination, with
$K(S)=\|(I-T)^{-1}\|$
measuring the worst-case amplification in solving $(I-T)y=b$. By contrast,
$\Gamma=(I-T)^{-1}B$
restricts this inverse to the particular right-hand sides $Bx$ generated by the public variables. Equivalently, $\Gamma$ is the back-substitution map that reconstructs the corresponding private component, $y=\Gamma x.$

The Schur determinant formula also gives
\be
\det(I-S)=\det(I-T)\det(I-V).
\ee
Since $I-T$ is assumed to be invertible, $\det(I-S)=0$ if and only if $\det(I-V)=0.$
Equivalently, from \eqref{a key formula of Schur} the eigenspaces of \(S\) and \(V\) corresponding to the eigenvalue \(1\) are naturally isomorphic via
\be
x \longmapsto \begin{pmatrix}x\\ \Gamma x\end{pmatrix}.
\ee
So $\dim\ker(I-S)=\dim\ker(I-V).$
This is also a standard feature of Schur-complement reduction: eliminating part of the variables retains the spectral information relevant to the remaining variables. In the present setting, since the reduction is applied to $I-S$, it preserves precisely the spectral information at eigenvalue $1$.

A particularly useful property of Schur complements is the quotient formula, which states that, whenever the relevant Schur complements are well defined, successive elimination of nested blocks agrees with eliminating those blocks at once~\cite{CrabtreeHaynsworth1969,zhang2005schur}. This gives a natural interpretation of the transducer composition in Lemma~\ref{lem:composition_transducer}.

Suppose
\be
S_j=
\begin{pmatrix}
D_j&E_j\\
B_j&T_j
\end{pmatrix}
\ee
is a transducer with action
$
V_j=D_j+E_j(I-T_j)^{-1}B_j.
$
In the composite construction, let \(q_j\) denote the state entering the public space of the \(j\)-th transducer, and let \(p_j\) denote its local private variable. They satisfy
\be
p_j=B_jq_j+T_jp_j,
\qquad
q_{j+1}=D_jq_j+E_jp_j.
\ee
Eliminating the local private variable \(p_j\) gives
\be
p_j=(I-T_j)^{-1}B_jq_j=\Gamma_jq_j,
\ee
and therefore
\be
q_{j+1}
=
\left[D_j+E_j(I-T_j)^{-1}B_j\right]q_j
=
V_jq_j.
\ee
Thus, after eliminating all local private variables, the composite transducer reduces to the clock propagation
$
q_{j+1}=V_jq_j,
$
or equivalently,
\be
q_1=V_0q_0,\qquad
q_2=V_1V_0q_0,\qquad
\ldots,\qquad
q_L=V_{L-1}\cdots V_0q_0.
\ee
One may then eliminate the intermediate clock variables \(q_1,\ldots,q_{L-1}\), which yields the overall action
$
V_{L-1}\cdots V_1V_0.
$
Equivalently, one may eliminate all local private variables and intermediate clock variables in a single Schur reduction. The quotient property of Schur complements guarantees that the two procedures give the same result.

In summary, the Schur-complement viewpoint gives the following dictionary:
$$
\begin{array}{c|c}
\text{Schur-complement language}
&
\text{Transducer language}
\\ \hline
\text{variables retained after elimination}
&
\text{public space}
\\
\text{variables eliminated}
&
\text{private space}
\\
\text{reduced operator}
&
I-V
\\
\text{back substitution}
&
\Gamma=(I-T)^{-1}B
\\
\text{inverse of the elimination pivot}
&
(I-T)^{-1}
\\
\text{worst-case amplification in elimination}
&
K(S)=\|(I-T)^{-1}\|
\\
\text{successive elimination}
&
\text{transducer composition}
\end{array}
$$

Therefore, at the algebraic level, a transducer may be viewed as a unitary realization of a Schur-complement reduction. The algorithmic content of the transducer framework lies not in the Schur-complement identity itself, but in constructing a larger unitary $S$ that is inexpensive to implement while the induced action $V$ on the public space may be substantially more expensive to realize directly. Equivalently, the Schur complement of $I-S$ with respect to the private block is $I-V$.
The clock constructions in
Lemmas~\ref{lem:composition_unitary} and
\ref{lem:composition_transducer}
illustrate this principle: they encode a sequential computation into the enlarged private space, so that eliminating the private variables recovers the desired composition on the public space.

\section{Transducer toolkit}\label{sec:transducer_toolkit}

In this section, we discuss some advanced transducer's propositions, which can be roughly categorized into two classes. 
The first part is the construction of a single transducer with action being the sequential composition of multiple unitaries, providing a methodology of transiting the query complexity brought by the sequential operations to the norms related to the private space of the transducer. 
The second part discusses how to turn a transducer to a block-encoding of its action without the need of explicit catalyst, formally called catalyst removal via finite reuse.

\subsection{Transducers for composition}\label{sec:transducer_toolkit_composition}

\begin{lemma}[Transducer for unitary composition]\label{lem:composition_unitary}
    Let $V = G_L O_{L-1} G_{L-1} \cdots G_1 O_0 G_0$ be a unitary circuit where $O_k$ represent one oracle call and $G_k$ are oracle-independent unitaries. 
    Then there exists a transducer $S$ with action $V$ such that: 
    \begin{enumerate}
        \item The transducer $S$ is in the form of 
        \begin{equation}
            S = \left( \begin{array}{cc}
                0 & E \\
                B & T
            \end{array} \right)
        \end{equation}
        with $T^L = 0$. 
        \item Each implementation of $S$ requires one call to the select oracle  $\sum_{k=0}^{L-1} \ket{k}\bra{k} \otimes O_k$. 
        \item We have the following identities and estimates: 
\be
\begin{alignedat}{2}
&V = E(I-T)^{-1}B,
\qquad &&K(S)=\|(I-T)^{-1}\|\le L,\\
&\|(I-T)^{-1}B\psi\|^2 = L\|\psi\|^2,
\qquad &&\|E(I-T)^{-1}\|=\sqrt{L}.
\end{alignedat}
\ee
    \end{enumerate}
\end{lemma}
\begin{proof}
    Denote $O_{-1} = I$ and $V_k = G_k O_{k-1} \cdots G_1 O_0 G_0$ (intuitively $V_k$ is a part of the circuit of $V$ truncated at $k$-th oracle calls, with $V_0 = G_0$ and $V_L = V$). 
    Let 
    \begin{align}
        S & = \sum_{k=0}^{L-1} \ket{k+1}\bra{k} \otimes G_k O_{k-1} + \ket{0}\bra{L} \otimes G_L O_{L-1} \nonumber \\
        & = \left( \begin{array}{cccccc}
            0 & & & & & G_{L} O_{L-1} \\
            G_0 & 0 & & & & \\
             & G_1 O_0 & 0 & & & \\
             & & G_2 O_1 & \ddots & & \\
             & & & \ddots & \ddots & \\
             & & &  & G_{L-1} O_{L-2} & 0 \\ 
        \end{array} \right). \label{eqn:proof_unitary_composition_def_S}
    \end{align}
    Then the unitary $S$ given in~\cref{eqn:proof_unitary_composition_def_S} is a transducer of $V$, because for any vector $\psi$, 
    \begin{align}
        & \quad S (\psi \oplus \left( V_0 \psi \oplus V_1 \psi \oplus \cdots \oplus V_{L-1} \psi \right) ) \\
        & = S \left( \ket{0}\otimes \psi + \sum_{k=1}^{L} \ket{k} \otimes V_{k-1} \psi \right) \\
        & = \ket{1}\otimes G_0 \psi + \sum_{k=1}^{L-1} \ket{k+1} \otimes G_k O_{k-1} V_{k-1} \psi + \ket{0}\otimes G_L O_{L-1} V_{L-1}\psi \\
        & = \ket{1}\otimes V_0 \psi + \sum_{k=2}^{L} \ket{k} \otimes V_{k-1} \psi + \ket{0}\otimes V_{L}\psi \\
        & = V \psi \oplus \left( V_0 \psi \oplus V_1 \psi \oplus \cdots \oplus V_{L-1} \psi \right). 
    \end{align}
    The corresponding catalyst map is 
    \begin{equation}
        \Gamma = \left(\begin{array}{c}
            V_0 \\
            V_1 \\
             \vdots \\
             V_{L-1}
        \end{array}\right), 
    \end{equation}
    and $E = [0,0,\cdots,0,G_{L}O_{L-1}]$, $B = [G_0;0;\cdots;0]$, and 
    \begin{equation}\label{eqn:proof_unitary_composition_def_T}
        T = \left( \begin{array}{ccccc}
             0 & & & & \\
             G_1 O_0 & 0 & & & \\
             & G_2 O_1 & \ddots & & \\
             & & \ddots & \ddots & \\
             & &  & G_{L-1} O_{L-2} & 0 \\ 
        \end{array} \right). 
    \end{equation}
    It is clear that $T^L = 0$ since $T$ is a lower triangular $L$-by-$L$ block matrix with all the diagonal blocks to be $0$. 
    Furthermore, notice that $S$ can be rewritten as 
    \begin{equation}
        S =  \left( \left( \sum_{k=0}^{L-1}\ket{k+1}\bra{k}+\ket{0}\bra{L} \right)\otimes I \right)\left( \sum_{k=0}^L \ket{k}\bra{k} \otimes G_k \right)\left(I \oplus \sum_{k=0}^{L-1} \ket{k}\bra{k} \otimes O_k\right), 
    \end{equation}
    so its implementation requires one query to the select oracle of $O_j$'s. 

    Now let us prove the identities and estimates claimed in the lemma. 
    The identity $V = E(I-T)^{-1} B$ directly follows from~\cref{lem:basic_prop_transducer} and noticing that the left upper block of $S$ here is $0$. 
    Since $T^L = 0$, we have
    \begin{equation}
        (I-T)^{-1} = I + T + T^2 + \cdots + T^{L-1},  
    \end{equation}
    and thus $K(S) = \|(I-T)^{-1}\| \leq L$. 
    A direct computation shows that\footnote{Here and below, \([A_1;\cdots;A_k]\) denotes the block column obtained by vertically stacking the vectors or matrices \(A_1,\ldots,A_k\).} 
    \begin{equation}
        (I-T)^{-1}B \psi = [V_0 \psi;V_1\psi;\cdots;V_{L-1}\psi], 
    \end{equation}
    so $\|(I-T)^{-1}B \psi\|^2 = \sum_{k=0}^{L-1} \|V_k \psi\|^2 = L \|\psi\|^2$. 
    For the last equation, by $(I-T)^{-1} = I+T+T^2+\cdots + T^{L-1}$, we can compute that 
\be
E(I-T)^{-1}
=
\begin{pmatrix}
R_0 & R_1 & \cdots & R_{L-2} & R_{L-1}
\end{pmatrix},
\ee
where $R_k=G_L O_{L-1}G_{L-1}O_{L-2}\cdots G_{k+1}O_k$ for $0\leq k\leq L-1$.
In particular, since each \(R_k\) is unitary,
$
E(I-T)^{-1}
\bigl(E(I-T)^{-1}\bigr)^\dagger
=
\sum_{k=0}^{L-1}R_kR_k^\dagger
=
LI,
$
and so
$
\|E(I-T)^{-1}\|=\sqrt{L}.
$
\end{proof}

In some scenarios, the intermediate unitaries in a circuit are realized through their transducers, so we also need a result for the composition of multiple transducers. 
We first discuss a simple case where all the unitaries are given by their transducers. 
Our construction is closely related to the parallel implementation of sequential transducers in~\cite[Proposition~9.10 and Theorem~10.5]{BelovsJefferyYolcu2024}. 
There, the intermediate states of a sequential composition are stored in the catalyst, which leads to the same additive bound on the catalyst norm as \cref{eqn:composition_transducer_w_bound}. 
Our result gives an explicit block-matrix realization of this construction and, in addition, derives the resolvent bound in \cref{eqn:composition_transducer_K_bound}, which will be essential for the finite-reuse implementation below. 

\begin{lemma}[Transducer for pure transducer composition]\label{lem:composition_transducer}
    For $0 \leq j \leq L-1$, let $V_j$ be unitaries of the same dimension, and $S_j=
\begin{pmatrix}
D_j&E_j\\
B_j&T_j
\end{pmatrix}$ be the transducer of $V_j$ with catalyst map $\Gamma_j$. 
    Then there exists a transducer $S$ of $V = V_{L-1}\cdots V_1V_0$ such that 
    \begin{align}
        w(S) &\leq L-1 + \sum_{j=0}^{L-1} w(S_j), \label{eqn:composition_transducer_w_bound}\\
        K(S) &\leq \max_j K(S_j) + (L-1) (1+\max_j w(S_j)), \label{eqn:composition_transducer_K_bound}
    \end{align}
    using one query to the select oracle of $S_j$. 
\end{lemma}

\begin{proof}
    We construct $S$ as\footnote{See Appendix \ref{app:B} for an explanation of this construction.}
    \begin{equation}\label{eqn:proof_transducer_composition_def_S}
        S = \left( \begin{array}{cccccccccc}
            0 & 0 & \cdots & 0 & D_{L-1} & 0 & 0 & \cdots & 0 & E_{L-1} \\
            D_0  & 0  & \cdots & 0 & 0 & E_0  & 0  & \cdots & 0 & 0 \\
            0 & D_1 & \cdots & 0 & 0 &  0 & E_1 & \cdots & 0 & 0 \\
            \vdots & \vdots & \ddots & \vdots & \vdots & \vdots & \vdots & \ddots & \vdots & \vdots \\
            0 & 0& \cdots & D_{L-2} & 0 & 0 & 0& \cdots & E_{L-2} & 0 \\
            B_0 & 0 & 0 & \cdots & 0 & T_0 & 0 & 0 & \cdots & 0 \\
            0 & B_1 & 0 & \cdots & 0 & 0 & T_1 & 0 & \cdots & 0 \\
            0 & 0 & B_2 & \cdots & 0 & 0 & 0 & T_2 & \cdots & 0 \\
            \vdots & \vdots & \vdots & \ddots & \vdots & \vdots & \vdots & \vdots & \ddots & \vdots \\
            0 & 0 & 0 & \cdots & B_{L-1} & 0 & 0 & 0 & \cdots & T_{L-1}
        \end{array} \right)
    \end{equation}
    Notice that $S$ can be constructed through one call to the select oracle of $S_j$ and permutations of the block order. 
    
    Now we show that $S$ is the desired transducer with the public space being the upper left block. 
    Let $\psi_0$ be an arbitrary vector of dimension the same as the column vector of $V$, and choose 
    \begin{equation}
        \psi_j = V_{j-1}V_{j-2}\cdots V_0 \psi_0, \quad j \geq 1. 
    \end{equation}
    Denote 
    \begin{equation}
        u = [ \psi_0 ; \psi_1; \cdots;\psi_{L-1}; \Gamma_0\psi_0;\Gamma_1\psi_1;\cdots; \Gamma_{L-1}\psi_{L-1} ] = \psi_0 \oplus \Gamma \psi_0, 
    \end{equation}
    where 
    \begin{equation}\label{eqn:proof_transducer_composition_def_Gamma}
        \Gamma = [ V_0; V_1V_0; \cdots; V_{L-2}\cdots V_1V_0; \Gamma_0;\Gamma_1V_0; \Gamma_2 V_1V_0; \cdots; \Gamma_{L-1} V_{L-2}\cdots V_1V_0 ]. 
    \end{equation}
    Then $Su = [x_0;x_1;\cdots;x_{L-1};y_0;y_1;\cdots; y_{L-1}]$, where 
    \begin{equation}
        x_{0} = D_{L-1} \psi_{L-1} + E_{L-1} \Gamma_{L-1} \psi_{L-1}, \quad x_{j+1} = D_j \psi_j + E_j \Gamma_j \psi_j, \quad 0 \leq j \leq L-2, 
    \end{equation}
    \begin{equation}
        y_j = B_j \psi_j + T_j \Gamma_j \psi_j, \quad 0 \leq j \leq L-1. 
    \end{equation}
    According to~\cref{lem:basic_prop_transducer}, we have $D_j + E_j \Gamma_j = V_j$ and $B_j + T_j \Gamma_j = \Gamma_j$, so 
    \begin{equation}
        x_0 = V_{L-1} \psi_{L-1}, \quad x_{j+1} = V_j \psi_j, \quad 0 \leq j \leq L-2, 
    \end{equation}
    \begin{equation}
        y_j = \Gamma_j \psi_j, \quad 0 \leq j \leq L-1. 
    \end{equation}
    From the iterative relation of $x_j$, we have $x_{0} = V_{L-1}\cdots V_1 V_0 \psi_0 = V \psi_0$ and $x_j = V_{j-1}\cdots V_1V_0 \psi_0$. 
    So we can write 
    \begin{align}
        Su & = S(\psi_0 \oplus \Gamma \psi_0) \\
        & = [ V\psi_0 ; V_0\psi_0; V_1 V_0 \psi_0; \cdots; V_{L-2} \cdots V_1 V_0 \psi_0; \Gamma_0 \psi_0; \Gamma_1 V_0 \psi_0; \Gamma_2 V_1 V_0 \psi_0 ; \cdots ; \Gamma_{L-1} V_{L-2} \cdots V_1 V_0 \psi_0 ] \\
        & = V\psi_0 \oplus \Gamma \psi_0, 
    \end{align}
    which verifies that $S$ defined in~\cref{eqn:proof_transducer_composition_def_S} is a transducer of $V$ with catalyst $\Gamma$ defined in~\cref{eqn:proof_transducer_composition_def_Gamma}. 

    Now let us bound the relevant norms of $S$. 
    From the definition of $\Gamma$ in~\cref{eqn:proof_transducer_composition_def_Gamma}, for any vector $\psi$, we have 
    \begin{align}
        \|\Gamma \psi\|^2 &= \|V_0 \psi\|^2 + \|V_1 V_0 \psi\|^2 + \cdots + \|V_{L-2}\cdots V_1 V_0 \psi\|^2 \\
        & \quad + \|\Gamma_0 \psi\|^2 + \|\Gamma_1 V_0 \psi\|^2 + \cdots + \|\Gamma_{L-1}V_{L-2}\cdots V_1 V_0 \psi\|^2 \\
        & \leq (L-1)\|\psi\|^2 + \sum_{j=0}^{L-1}\|\Gamma_j\|^2 \|\psi\|^2, 
    \end{align}
    thus 
    \begin{equation}
        w(S) = \|\Gamma\|^2 \leq L-1 + \sum_{j=0}^{L-1} w(S_j). 
    \end{equation}

    To estimate $K(S) = \|(I-T)^{-1}\|$, let us solve the linear system 
    \begin{equation}
        (I-T) \left( \begin{array}{c}
            h_1 \\
            h_2 \\
            \vdots \\
            h_{L-1} \\
            v_0 \\
            v_1 \\
            v_2 \\
            \vdots \\
            v_{L-1}
        \end{array} \right) = \left( \begin{array}{c}
            a_1 \\
            a_2 \\
            \vdots \\
            a_{L-1} \\
            b_0 \\
            b_1 \\
            b_2 \\
            \vdots \\
            b_{L-1}
        \end{array} \right) =: g
    \end{equation}
    for arbitrary vectors $a_j$ and $b_j$. 
    Then we have the recursive relation (with supplementary definition $h_0 = 0$)
    \begin{align}
        h_j - D_{j-1} h_{j-1} - E_{j-1} v_{j-1} & = a_j \label{eqn:proof_transducer_composition_rec_eq1}\\
        -B_j h_j + (I-T_j)v_j & = b_j. \label{eqn:proof_transducer_composition_rec_eq2}
    \end{align}
    Using~\cref{lem:basic_prop_transducer}, \cref{eqn:proof_transducer_composition_rec_eq2} becomes 
    \begin{equation}\label{eqn:proof_transducer_composition_rec_eq3}
        v_j = \Gamma_j h_j + (I-T_j)^{-1} b_j. 
    \end{equation}
    Substituting this into~\cref{eqn:proof_transducer_composition_rec_eq1} gives 
    \begin{align}
        h_j &= D_{j-1} h_{j-1} + E_{j-1} (\Gamma_{j-1} h_{j-1} + (I-T_{j-1})^{-1} b_{j-1}) + a_j \\
        & = (D_{j-1} + E_{j-1} \Gamma_{j-1}) h_{j-1} + a_j + E_{j-1} (I-T_{j-1})^{-1} b_{j-1}. 
    \end{align}
    Therefore, according to~\cref{lem:basic_prop_transducer} again, 
    \begin{equation}
        h_j = V_{j-1} h_{j-1} + a_j + V_{j-1} \Gamma_{j-1}^{\dagger} b_{j-1}. 
    \end{equation}
    Solving this recursive relation gives a block matrix form $h = \mathcal{L} f$, where $h = [h_1;h_2;\cdots;h_{L-1}]$, $f = [f_1;f_2;\cdots;f_{L-1}]$ with 
    \begin{equation}
         f_j = a_j + V_{j-1} \Gamma_{j-1}^{\dagger} b_{j-1}, 
    \end{equation}
    and 
    \begin{equation}
        \mathcal{L} = \left( \begin{array}{ccccc}
            I & 0 & 0 & \cdots & 0 \\
            V_1 & I & 0 & \cdots & 0 \\
            V_2 V_1 & V_2 & I & \cdots & 0 \\
            \vdots & \vdots & \vdots & \ddots & \vdots \\
            V_{L-2}\cdots V_1 & V_{L-2}\cdots V_2 & \cdots & V_{L-2} & I 
        \end{array} \right). 
    \end{equation}
    Notice that $\|\mathcal{L}\| \leq L-1$ and $\|f_j\| \leq \|a_j\| + \|\Gamma_{j-1}\| \|b_{j-1}\|$, we can bound 
    \begin{align}
        \|h\|^2 & \leq \|\mathcal{L}\|^2 \|f\|^2 = \|\mathcal{L}\|^2 \sum_{j=1}^{L-1} \|f_j\|^2 \\
        & \leq (L-1)^2 \sum_{j=1}^{L-1} ( \|a_j\| + \|\Gamma_{j-1}\| \|b_{j-1}\| )^2 \\
        & \leq (L-1)^2 \sum_{j=1}^{L-1} (1+w(S_{j-1})) ( \|a_j\|^2 + \|b_{j-1}\|^2 ) \\
        & \leq (L-1)^2 (1+\max_j w(S_{j})) \|g\|^2. 
    \end{align}
    Now, using~\cref{eqn:proof_transducer_composition_rec_eq3}, we decompose 
    \begin{equation}\label{eqn:proof_transducer_composition_rec_eq4}
         (I-T)^{-1}g = \left( \begin{array}{c}
            h_1 \\
            \vdots \\
            h_{L-1} \\
            v_0 \\
            v_1 \\
            \vdots \\
            v_{L-1}
        \end{array} \right) = \left( \begin{array}{c}
            0 \\
            \vdots \\
            0 \\
            (I-T_0)^{-1} b_0 \\
            (I-T_1)^{-1} b_1 \\
            \vdots \\
            (I-T_{L-1})^{-1} b_{L-1}
        \end{array} \right) + \left( \begin{array}{c}
            h_1 \\
            \vdots \\
            h_{L-1} \\
            0 \\
            \Gamma_1 h_1 \\
            \vdots \\
            \Gamma_{L-1} h_{L-1}
        \end{array} \right). 
    \end{equation}
    Then the second term on the right hand side of~\cref{eqn:proof_transducer_composition_rec_eq4} has norm 
    \begin{align}
        \sqrt{ \|h\|^2 + \sum_{j=1}^{L-1} \| \Gamma_j h_j\|^2 } & \leq \sqrt{ \|h\|^2 + \max_j w(S_j) \sum_{j=1}^{L-1} \| h_j\|^2 } \\
        & = \sqrt{1+\max_j w(S_j)} \|h\| \\
        & \leq (L-1) (1+\max_j w(S_j)) \|g\|. 
    \end{align}
    For the first term on the right hand side of~\cref{eqn:proof_transducer_composition_rec_eq4}, its norm is 
    \begin{equation}
        \sqrt{\sum_{j=0}^{L-1} \|(I-T_j)^{-1} b_j \|^2} \leq \max_j \|(I-T_j)^{-1}\| \sqrt{ \sum_{j=0}^{L-1} \| b_j \|^2 } \leq \max_j K(S_j) \|g\|. 
    \end{equation}
    Combining these two estimates together, we have 
    \begin{equation}
        \|(I-T)^{-1} g\| \leq ( \max_j K(S_j) + (L-1) (1+\max_j w(S_j)) ) \|g\|, 
    \end{equation}
    and thus 
    \begin{equation}
        K(S) = \| (I-T)^{-1} \| \leq \max_j K(S_j) + (L-1) (1+\max_j w(S_j)). 
    \end{equation}
\end{proof}

\Cref{lem:composition_transducer} shows how to construct a larger transducer for the sequential composition of multiple transducers. 
\cref{lem:composition_transducer} also estimates two norms. 
Importantly,~\cref{eqn:composition_transducer_K_bound} shows that the individual resolvent contributions enter through \(\max_j K(S_j)\), rather than through their sum.
Later we will see that $K(S)$ is directly related to the algorithm complexity, so~\cref{eqn:composition_transducer_K_bound} is the key reason why transducer-based LCU can outperform the standard LCU implementation. 

Although~\cref{lem:composition_transducer} only considers the case where all the unitaries $V_j$ are realized through the transducers, the same result can be generalized to allow inverse, controlled version, untouched qubits, and transducer-independent gates. 
Specifically, if the circuit includes $V_j^{\dagger}$, then we can simply use $S_j^{\dagger}$. 
To see this, for any vector $\psi$, we define $\phi = V_j^{\dagger} \psi$ and using the definition of the transducer for $\phi$ to write 
\begin{equation}
    S_j (\phi \oplus \Gamma_j \phi) = V_j \phi \oplus \Gamma_j \phi. 
\end{equation}
Apply $S_j^{\dagger}$ on both sides, and we have 
\begin{equation}
    S_j^{\dagger} (\psi \oplus \Gamma_j V_j^{\dagger} \psi ) = V_j^{\dagger} \psi \oplus \Gamma_j V_j^{\dagger} \psi, 
\end{equation}
i.e., $S_j^{\dagger}$ is a transducer of $V_j^{\dagger}$ with the norm $w(S_j) = \|\Gamma_j V_j^{\dagger}\|^2 = \|\Gamma_j\|^2$ unchanged. 
Furthermore, $S_j^{\dagger} = [D_j^{\dagger},B_j^{\dagger};E_j^{\dagger},T_j^{\dagger}]$, so $K(S_j^{\dagger}) = \|(I-T_j^{\dagger})^{-1}\| = \|(I-T_j)^{-1}\| = K(S_j)$ also remains unchanged. 
Another scenario is that there are additional $V_j$-independent gates. 
For example, if we would like to implement $G_j V_j$, where $G_j$ represents a collection of known gates, then we can use $[G_j,0;0,I] S_j$, which is the transducer of $G_jV_j$ with catalyst and the lower right block unchanged. 
If $G_j$ comes before $V_j$, i.e., $ V_jG_j$, then we use $S_j [G_j,0;0,I] $, which transduces $V_jG_j $ with the same lower right block. 
The corresponding catalyst becomes $\Gamma_j G_j$, so its norm remains unchanged. 
For the controlled version of $V_j$, we regard the control register as
part of the public space. Writing
\[
S_j=
\begin{pmatrix}
D_j&E_j\\
B_j&T_j
\end{pmatrix},
\]
we may use
\[
S_j^{\rm ct}
=
\begin{pmatrix}
I&0&0\\
0&D_j&E_j\\
0&B_j&T_j
\end{pmatrix},
\]
which transduces the controlled-$V_j$ operation. Its private block
remains $T_j$, and hence
\[
w(S_j^{\rm ct})=w(S_j),
\qquad
K(S_j^{\rm ct})=K(S_j).
\]

Summarizing the above discussions, we have the following result. 

\begin{corollary}[Transducer for transducer and unitary composition]\label{cor:composition_general}
    For $1 \leq j \leq L$, let $V_j$ be unitaries (not necessarily of the same dimension), and $S_j$ be the transducer of $V_j$ with catalyst map $\Gamma_j$. 
    Let $\widetilde V_j$, unitaries of the same dimension, represent one of the following cases: $V_j$ itself or $V_j\otimes I$, their inverses and controlled versions.
    Let $G_j$ denote unitaries independent of $V_j$'s. 
    Then there exists a transducer $S$ of $G_L\widetilde V_LG_{L-1}\widetilde V_{L-1}\cdots G_1\widetilde V_1G_0$ such that 
    \begin{align}
        w(S) &\leq L-1 + \sum_{j=1}^{L} w(S_j), \\
        K(S) &\leq \max_j K(S_j) + (L-1) (1+\max_j w(S_j)),
    \end{align}
    using one query to the select oracle of $S_j$. 
\end{corollary}

\subsection{Catalyst removal via finite reuse}\label{sec:transducer_toolkit_finite_reuse}

We now explain the finite reuse construction that will be used to turn a transducer into a unitary implementation of its action without the need of explicit catalyst. 
The idea is to introduce $N$ copies of the public space and a shared private space initialized by a zero vector input rather than the exact catalyst. 
The finite reuse then repeats the transducer for $N$ times with different public spaces and the shared same private space, reusing the private output of one as the private input of the next transducer operation. 

We discuss the finite reuse procedure more precisely. 
Let $S = [ D,E;B,T ]$ be a unitary transducer with action $V$ and catalyst $\Gamma$. 
Introduce $\mathcal H_{\rm pub}^{(0)},\ldots, \mathcal H_{\rm pub}^{(N-1)} \simeq \mathcal H_{\rm pub}$, and the finite reuse circuit acts on $\left( \bigoplus_{j=0}^{N-1} \mathcal H_{\rm pub}^{(j)} \right) \oplus \mathcal H_{\rm priv}$. 
For an arbitrary input vector $\psi$, we first append ancilla qubits to get $(\ket{0}\otimes \psi) \oplus \mathbf{0}$, where the first register consists of $\log N$ ancilla qubits, and $\mathbf{0} = [0;0;\cdots;0]$ represents the zero vector. 
We apply the Hadamard gates to prepare a uniform superposition over the first ancilla register, and the state becomes 
\begin{equation}
    \left(\frac{1}{\sqrt{N}} \sum_{j=0}^{N-1} \ket{j} \otimes \psi\right) \oplus \mathbf{0} = [\psi/\sqrt{N};\cdots;\psi/\sqrt{N};\mathbf{0}] \in \left( \bigoplus_{j=0}^{N-1} \mathcal H_{\rm pub}^{(j)} \right) \oplus \mathcal H_{\rm priv}. 
\end{equation}

Then we apply $S$ for $N$ times, sequentially on the space $\mathcal H_{\rm pub}^{(j)} \oplus \mathcal H_{\rm priv}$ for $0 \leq j \leq N-1$. 
To describe this step explicitly, let $q_j$ denote the private vector immediately before the $j$-th application of
$S$ (and thus $q_0 = \mathbf{0}$), and let $o_j$ denote the public output produced after that application. 
At each reuse, the public input is always $\psi/\sqrt N$, so we have 
\begin{equation}
S \begin{pmatrix}
	\psi/\sqrt N\\
	q_j
\end{pmatrix} =
\begin{pmatrix}
	o_j\\
	q_{j+1}
\end{pmatrix}, 
\end{equation} 
i.e., 
\begin{equation}
	\label{eq:finite-reuse-recursion}
	o_j = \frac{1}{\sqrt N}D\psi+Eq_j,
	\quad
	q_{j+1}
	=
	\frac{1}{\sqrt N}B\psi+T q_j,
	\quad
	0\leq j \leq N-1.
\end{equation}
Notice that both $o_j$ and $q_j$ can be written as a linear transformation of $\psi$. 
After all the $N$ applications of $S$, the entire output vector becomes 
\begin{equation}
    \left( \sum_{j=0}^{N-1} \ket{j} \otimes o_j \right) \oplus q_N. 
\end{equation}
Applying Hadamard gates on the first register again yields the output like 
\begin{equation}
    \left( \ket{0}\otimes P_N \psi + \sum_{j=1}^{N-1} \ket{j}\otimes * \right) \oplus *, 
\end{equation}
where
\begin{equation}
	\label{eq:def-PN}
	P_N\psi := \frac{1}{\sqrt N} \sum_{j=0}^{N-1}o_j.
\end{equation}
Post-selecting all the ancilla qubits (including those in the first register and for private space) onto $0$ gives $P_N\psi$. 
Therefore, the pre-described circuit constructs a block-encoding of $P_N$. 

Since the private space is not the ideal catalyst, $P_N$ is thereby not the ideal action $V$. 
The following lemma identifies the error incurred by starting
with zero private input. 

\begin{lemma}[Reuse error identity \cite{chen2026optimal}]\label{lem:finite-reuse}
	Let $S$ be a transducer with action $V$ and catalyst $\Gamma$. 
    Let $P_N$ be the upper left block of the length-$N$ finite reuse circuit, defined by \cref{eq:def-PN}. 
    Then
	\begin{equation}
	V - P_N = E G_N(T)\Gamma, \quad G_N(z) = \frac1N \sum_{k=0}^{N-1} z^k.
	\end{equation}
\end{lemma}

\begin{proof}
    Recall that for any input vector $\psi$, $P_N \psi = (\sum_{j=0}^{N-1} o_j) / \sqrt{N}$ where $o_j$'s are given by the recursion~\cref{eq:finite-reuse-recursion}. 
    Solving the linear recursion for $q_j$ with $q_0 = \mathbf{0}$ yields 
    \begin{equation}
        q_j = \frac{1}{\sqrt{N}}(I+T+T^2+\cdots + T^{j-1}) B \psi. 
    \end{equation}
    According to~\cref{lem:basic_prop_transducer}, we have $B = (I-T)\Gamma$, and thus 
    \begin{align}
        q_j = \frac{1}{\sqrt{N}}(I+T+T^2+\cdots + T^{j-1}) (I-T)\Gamma  \psi = \frac{1}{\sqrt{N}} (I - T^j) \Gamma \psi. 
    \end{align}
    Substituting it into the first equation of~\cref{eq:finite-reuse-recursion} gives 
	\begin{equation}
		o_j = \frac{1}{\sqrt N}D\psi + E \frac{1}{\sqrt N}(I-T^j)\Gamma\psi = \frac{1}{\sqrt N} ( D+E\Gamma-E T^j \Gamma )\psi = \frac{1}{\sqrt N} ( V-E T^j \Gamma )\psi, 
	\end{equation}
    where the last equation follows from~\cref{lem:basic_prop_transducer} again. 
    Therefore 
    \begin{equation}
        P_N\psi = \frac{1}{\sqrt{N}} \sum_{j=0}^{N-1} o_j = \frac{1}{N} \sum_{j=0}^{N-1} ( V-E T^j \Gamma )\psi = V \psi - E G_N(T) \Gamma \psi, 
    \end{equation}
    and thus 
    \begin{equation}
        V - P_N = E G_N(T) \Gamma. 
    \end{equation}
\end{proof}

\cref{lem:finite-reuse} already suggests that the finite reuse can turns a transducer to a block-encoding of $\epsilon$-approximation of $V$ for arbitrarily small $\epsilon$, but the convergence is only first-order. 
Specifically, suppose that $\|T\| < 1$, then $\|G_N(T)\| = \mathcal{O}(1/N)$, and \cref{lem:finite-reuse} implies $\|V - P_N\| = \mathcal{O}(1/N)$. 
So it suffices to choose $N = \mathcal{O}(1/\epsilon)$ for an $\epsilon$-approximation. 

To achieve high-order convergence, we can choose different values of $N$ and linearly combine the corresponding block-encodings to cancel the low-order error terms. 
This is in principle the extrapolation method, which has been widely used for improved accuracy in both classical and quantum numerical algorithms. 
Classical examples include Richardson extrapolation for finite-difference discretisation \cite{richardson1910} and extrapolation methods for ordinary differential equations such as the Bulirsch--Stoer method \cite{bulirsch1966numerical}. 
A closely related idea also appears in quantum Hamiltonian simulation through multi-product formulas, where linear combinations of product formula approximations with different step sizes are chosen to cancel low-order Trotter errors and thereby achieve higher order accuracy \cite{low2019wellconditionedmultiproducthamiltoniansimulation,aftab2026multiproducthamiltoniansimulationexplicit}, and the reduced-depth quantum algorithms for Hamiltonian simulation~\cite{WatsonWatkins2024} and linear differential equations~\cite{fang2026circuitdepthreductiononeancilla}. 
Similar idea further serves as the guiding principle for the famous zero-noise extrapolation~\cite{TemmeBravyiGambetta2017,LiBenjamin2017} in quantum error mitigation.

The idea of combining reuse circuits of different lengths was recently introduced in the context of optimal time-dependent Hamiltonian simulation \cite{chen2026optimal,chen2026gateefficientimplementationqueryoptimaltimedependent}, where a carefully chosen weighted combination exploits the time-ordered structure of the private block to achieve factorial error decay. 
The same principle has subsequently been applied to Lindblad simulation \cite{wang2026queryoptimalquantumsimulationlindblad}. 
Here we establish a structure-independent version of this idea. 
Rather than relying on a special form of the private block, we show that the sole assumption that $I-T$ is invertible is sufficient. 
We formally state the result as follows.

\begin{theorem}[High-order finite reuse]\label{thm:finite-reuse-high-order}
    Let $S = [D,E; B,T]$ be a transducer of $V$ with invertible $(I-T)$. 
    Then, we can implement a block-encoding of $\widetilde{V}$ such that $\|\widetilde{V} - V\| \leq \mathcal{O}(\epsilon)$, using 
    \begin{equation}
        \mathcal{O}\left( K(S) \log\left( \frac{1}{\epsilon} \right) \right)
    \end{equation}
    queries to $S$, its inverse, and controlled versions. 
\end{theorem}
\begin{proof}

    Following the notation in~\cref{lem:finite-reuse}, let $P_N$ denote the matrix constructed through the length-$N$ finite reuse circuit. 
    We consider the linear combination of finite reuse circuits with different values of $N$ as 
    \begin{equation}\label{eqn:proof_high_order_finite_reuse_def_Vt}
        \widetilde{V} = \sum_{k=0}^{L} b_k \left( \frac{R_0+k}{R_0} P_{R_0+k} - \frac{k}{R_0} P_k \right). 
    \end{equation}
    Here $b_k$, $L$ and $R_0$ are some parameters to be determined later, but we first suppose that $b_k \geq 0$ and $\sum b_k = 1$. 
    For any fixed $k$, according to~\cref{lem:finite-reuse}, 
    \begin{align}
        \frac{R_0+k}{R_0} P_{R_0+k} - \frac{k}{R_0} P_k &= \frac{R_0+k}{R_0} ( V - E G_{R_0+k}(T) \Gamma )- \frac{k}{R_0} ( V - E G_{k}(T) \Gamma ) \\
        & = V - E \left( \frac{R_0+k}{R_0} G_{R_0+k}(T) - \frac{k}{R_0}G_{k}(T) \right) \Gamma \\
        & = V - E \left( \frac{1}{R_0} \sum_{j=k}^{R_0+k-1} T^j \right) \Gamma \\
        &= V - E T^k G_{R_0}(T) \Gamma. 
    \end{align}
    Then
    \begin{equation}
        \widetilde{V} = \sum_{k=0}^{L} b_k V - \sum_{k=0}^{L} b_k E T^k G_{R_0}(T) \Gamma  = V - E \Phi(T) G_{R_0}(T) \Gamma, 
    \end{equation}
    where we have used $\sum b_k = 1$ and denote 
    \begin{equation}
        \Phi(z) = \sum_{k=0}^L b_k z^k. 
    \end{equation}
    
    To bound this error, we first rewrite 
    \begin{equation}
        G_{R_0}(T) = \frac{1}{R_0} \sum_{k=0}^{R_0-1} T^k = \frac{1}{R_0} (I-T)^{-1} (I - T^{R_0}). 
    \end{equation}
    Notice that the matrices $(I-T)^{-1}, \Phi(T), I - T^{R_0}$ commute with one another. 
    Hence 
    \begin{align}
        V - \widetilde{V} &= \frac{1}{R_0}  E \Phi(T)(I-T)^{-1} (I - T^{R_0}) \Gamma \\
        &= \frac{1}{R_0}  E (I-T)^{-1} (I - T^{R_0}) \Phi(T) \Gamma \\
        & = \frac{1}{R_0}  V \Gamma^{\dagger} (I - T^{R_0}) \Phi(T) \Gamma
    \end{align}
    where the last line follows from~\cref{lem:basic_prop_transducer}. 
    Taking the norm, we have 
    \begin{equation}\label{eqn:proof_high_order_finite_reuse_error}
        \|V - \widetilde{V}\| \leq \frac{w(S)}{R_0} \| I - T^{R_0} \| \|\Phi(T)\|. 
    \end{equation}

    In order to control the error, we need to construct proper $\Phi(T)$ with sufficiently small norm. 
    Notice that $\Phi(1) = 1$ due to $\sum b_k = 1$, so we have to bound $T$ away from $1$, which, to some extent, can be achieved by leveraging the invertible assumption of $(I-T)$. 
    Motivated by this, the idea of the construction is to choose\footnote{Notice that $G_{N_0}(1) = 1$ and thus $\Phi(1) = 1$ for any $N_0$ and $q$, and all the coefficients of $G_{N_0}(z)$ are positive. Hence the conditions $b_k \geq 0$ and $\sum b_k = 1$ are automatically satisfied.}
    \begin{equation}
        \Phi(z) = (G_{N_0}(z))^q
    \end{equation}
    for some integers $N_0$ and $q$ (so the degree $L = q(N_0-1)$), select $N_0$ such that $\|G_{N_0}(T)\|$ (which depends on $\|(I-T)^{-1}\|$) is strictly smaller than $1$, and choose $q$ to be logarithmic in $1/\epsilon$ to achieve the desired accuracy. 
    Specifically, using $G_{N_0}(T) = \frac{1}{N_0} (I-T)^{-1} (I - T^{N_0})$ again, and noticing that $\|T\| \leq 1$ since it is a sub-block of the unitary $S$, we have 
    \begin{equation}
        \|G_{N_0}(T)\| \leq \frac{2}{N_0} K(S) \leq \frac{1}{2}, 
    \end{equation}
    by choose 
    \begin{equation}
        N_0 = \lceil 4K(S) \rceil. 
    \end{equation}
    Then $\|\Phi(T)\| = \|(G_{N_0}(T))^q\| \leq 2^{-q}$, and thus~\cref{eqn:proof_high_order_finite_reuse_error} becomes 
    \begin{equation}
        \|V - \widetilde{V}\| \leq \frac{2 w(S)}{R_0} \frac{1}{2^q}. 
    \end{equation}
    To force the right hand side to be smaller than $\epsilon$, it suffices to choose 
    \begin{equation}\label{eqn:proof_high_order_finite_reuse_parameter_choice}
        q = \left\lceil \log_2 \left( \frac{1}{\epsilon}\right) \right\rceil, \quad  L = q(N_0-1) = \left\lceil \log_2 \left( \frac{1}{\epsilon}\right) \right\rceil \left( \lceil 4K(S) \rceil - 1 \right), \quad R_0 = 2\left\lceil \log_2 \left( \frac{1}{\epsilon}\right) \right\rceil \lceil 4K(S) \rceil . 
    \end{equation}
    Then $2^{-q} \leq \epsilon$, and $R_0 \geq 2(2 K(S) - 1) \geq 2w(S)$ according to~\cref{lem:basic_prop_transducer}, which together imply $\|V - \widetilde{V}\| \leq \epsilon$. 
    Notice that $R_0$ does not need to be as large as $\mathcal{O}(L)$ if we only wish to bound the error by $\epsilon$, but such a larger choice is necessary to ensure the normalization factor of the LCU implementation bounded, as we will discussed next. 

    To implement~\cref{eqn:proof_high_order_finite_reuse_def_Vt}, we can directly apply standard LCU in \Cref{lem:standard_lcu}, which gives a block-encoding of $\widetilde{V}/\alpha$ with normalization factor 
    \begin{equation}
        \alpha = \sum_{k=0}^L b_k \left( \frac{R_0+k}{R_0} + \frac{k}{R_0} \right) = \sum_{k=0}^L b_k \frac{R_0+2k}{R_0} \leq 1 + \frac{2L}{R_0} < 2,  
    \end{equation}
    where the last inequality is $R_0 \geq 2L$, guaranteed by the choice of $R_0$ and $L$ in~\cref{eqn:proof_high_order_finite_reuse_parameter_choice}. 
    We relax this normalization factor to exactly $2$ by append an additional ancilla qubit and perform a rotation gate. 
    Then, we apply the robust oblivious amplitude amplification technique~\cite{BerryChildsCleveEtAl2015} using $3$ calls to the LCU circuit to boost the normalization factor to $1$ with block-encoding error being $\mathcal{O}(\epsilon)$. 
    
    This procedure takes $\mathcal{O}(1)$ queries to the select oracle of $P_j$'s block-encodings. 
    As, according to~\cref{lem:finite-reuse}, the block-encoding of $P_j$ is constructed mainly by the sequentially application of a single transducer $S$ for $j$ times, the select oracle of $P_j$'s block-encodings can be constructed by imposing proper control logic on the most expensive $P_j$, which is $P_{R_0+L}$. 
    This requires $\mathcal{O}(R_0+L)$ calls to $S$, which yields the desired complexity according to the choice of $R_0$ and $L$ in~\cref{eqn:proof_high_order_finite_reuse_parameter_choice}. 
\end{proof}

\begin{remark}[Comparison with~\cref{lem:trans-to-alg}]
\label{rmk:finite-reuse-comparison}
For a fixed input $\psi$, the standard transducer implementation achieves state error $\epsilon$
using
\be
    \mathcal{O}\left(
        1+\frac{W(S,\psi)}{\epsilon^2}
    \right)
\ee
controlled applications of $S$. In particular, using
$W(S,\psi)\leq w(S)$ gives a worst-case bound
$\mathcal{O}(1+w(S)/\epsilon^2)$.
By contrast, \cref{thm:finite-reuse-high-order} gives an
operator-norm approximation to the entire public action using
\be
    \mathcal{O}\left(
        K(S)\log\frac1\epsilon
    \right)
\ee
queries, provided a bound on $K(S)=\|(I-T)^{-1}\|$ is available.
Thus, for a fixed transducer with controlled resolvent norm, the
dependence on the target precision improves from inverse polynomial
to logarithmic. This is not a uniform improvement in the transducer
parameter: by \cref{lem:basic_prop_transducer}, we have $w(S)\leq 2K(S)-1$,
and $K(S)$ can in general be substantially larger than $w(S)$.
The advantage of the high-order construction is therefore most useful
when the resolvent norm can be bounded comparably to, or only mildly
worse than, the catalyst complexity, as occurs in the constructions
used in this work.
\end{remark}

\section{Transducer-based linear combination of unitaries algorithm}\label{sec:transducer_LCU}

Now we are ready to establish our transducer-based LCU algorithm. 
We first discuss an important subroutine that constructs a transducer of the Cayley transformation of LCU with normalized coefficients. 
Then we present the entire transducer-based LCU algorithm and analyze its complexity.

\subsection{Cayley transformation of LCU with normalized coefficients}\label{sec:transducer_LCU_normalized}

In this subsection, let $V_j$ be a Hermitian block-encoding of $H_j$, and suppose that $V_j$ satisfies~\cref{assump:U}. 
We consider the linear combination 
\begin{equation}
    Y = \sum_{j=0}^{J-1} q_j H_j, 
\end{equation}
where we assume $q_j \geq 0$ and $\sum q_j = 1$. 
Let $V$ be the select oracle of $V_j$, i.e., 
\begin{equation}
    V = \sum_{j=0}^{J-1} \ket{j}\bra{j} \otimes V_j, 
\end{equation}
and $U_p$ prepares the coefficient state in the sense that 
\begin{equation}
    U_p \ket{0} = \sum_{j=0}^{J-1} \sqrt{q_j} \ket{j}. 
\end{equation}

The subroutine is as follows. 
The full space consists of four registers: the system register which is used for the data, the signal register which serves as the ancilla register of the block-encoding, the index register for $j$ in the linear combination, and the clock register introduced by the transducer. 
We first apply~\cref{lem:composition_unitary} to construct a transducers $S_{V_j} = [0, E_{V_j}; B_{V_j}, T_{V_j}]$ of $V_j$, then consider $\oplus S_{V_j}$ with reordering all the public sectors first then followed by private sections to get $S_V$. 
Specifically, $S_V = [0,E_{V}; B_V, T_V]$ where $E_V = \diag(E_{V_0},\cdots,E_{V_{J-1}})$, $B_V = \diag(B_{V_0},\cdots,B_{V_{J-1}})$, and $T_V = \diag(T_{V_0},\cdots,T_{V_{J-1}})$. 
Then we construct 
\begin{equation}\label{eqn:def_LCU_Cayley_transducer}
    S = \Phi \mathcal{U}_p^{\dagger} S_V \mathcal{U}_p. 
\end{equation}
Here 
\begin{equation}
    \mathcal{U}_p = \ket{0}\bra{0} \otimes U_p \otimes I_{\operatorname{sig,sys}} + \sum_{k \neq 0} \ket{k}\bra{k} \otimes I_{\operatorname{ind,sig,sys}}, 
\end{equation}
and, using the same block division as in $S_V$, 
\begin{equation}
    \Phi = \left( \begin{array}{cc}
        \Phi_0 & 0 \\
        0 & I
    \end{array} \right), \quad \Phi_0 = i (I - 2 \ket{0}\bra{0}\otimes \ket{0}\bra{0} \otimes I). 
\end{equation}
Roughly speaking, $\Phi$ appends $-i$ when all the ancilla registers are $0$, $i$ when the clock register is zero but the index and signal are not all zero, and does not append any factor otherwise. 
If we declare the space with all zero ancilla registers (but excluding the ancilla) to be the public space, then~\cref{eqn:def_LCU_Cayley_transducer} is a transducer of $\operatorname{Cay}(Y)$, as shown in the next result. 
Here the Cayley transformation is defined as 
\begin{equation}
    \operatorname{Cay}(Y) = (I - iY) (I + iY)^{-1}. 
\end{equation}

\begin{theorem}[Cayley transformation of LCU]\label{thm:Cayley_LCU}
     Let $V_j$ be a Hermitian block-encoding of $H_j$, and suppose that $V_j$ satisfies~\cref{assump:U} with $C_j$ oracle queries. 
     Let $Y = \sum_{j=0}^{J-1} q_j H_j$ with $q_j \geq 0$ and $\sum q_j = 1$. 
     Then, the circuit $S$ defined in~\cref{eqn:def_LCU_Cayley_transducer} gives a transducer of $\operatorname{Cay}(Y)$ with catalyst $\Gamma$, such that 
     \begin{equation}
         \Gamma^{\dagger}\Gamma = 2(1+\sum q_j C_j) (I+Y^2)^{-1} - I, 
     \end{equation}
     \begin{equation}
         w(S) \leq 1 + 2\sum q_j C_j, \quad K(S) \leq 6 \max C_j, 
     \end{equation}
     using $\mathcal{O}(1)$ queries to the primitive oracles. 
\end{theorem}

\begin{proof}

    We first define some notations which will be used through this proof. 
    Suppose the dimension of $H_j$ is $d$. 
    Let 
    \begin{equation}
        V = \sum_{j=0}^{J-1} \ket{j}\bra{j} \otimes V_j,
    \end{equation}
    and 
    \begin{equation}
        \mathcal{J}_0 = \ket{0}_{\text{sig}} \otimes I_{\text{sys}}, \quad \mathcal{J} = \sum_{j=0}^{J-1} \sqrt{q_j} \ket{j} \otimes \mathcal J_0 = (U_p\otimes I)(\ket{0}\otimes \mathcal J_0). 
    \end{equation}
    Then 
    \begin{equation}
        (U_p \otimes I_{\text{sig,sys}}) ( \ket{0}\ket{0} \otimes I ) = (U_p\ket{0})\ket{0}\otimes I = \mathcal{J}, 
    \end{equation}
    i.e., the first $d$ column of $U_p \otimes I_{\text{sig,sys}}$ is $\mathcal{J}$. 
    Then we can write the full representation of $U_p \otimes I_{\text{sig,sys}}$ as 
    \begin{equation}
        U_p \otimes I_{\text{sig,sys}} = (\begin{array}{cc}
            \mathcal{J} & \mathcal{J}_{\perp}
        \end{array}), 
    \end{equation}
    where $\mathcal{J}^{\dagger} \mathcal{J}_{\perp} = 0, \mathcal{J}_{\perp}^{\dagger} \mathcal{J}_{\perp} = I$. 
    Additionally, by the standard LCU calculation, we have 
    \begin{equation}
        \mathcal{J}^{\dagger} V \mathcal{J} = Y. 
    \end{equation} 

    The algorithm first applies~\cref{lem:composition_unitary} to construct transducers $S_{V_j} = [0, E_{V_j}; B_{V_j}, T_{V_j}]$ of $V_j$ using $\mathcal{O}(1)$ queries, then construct $S_V = [0,E_{V}; B_V, T_V]$ where $E_V = \diag(E_{V_0},\cdots,E_{V_{J-1}})$, $B_V = \diag(B_{V_0},\cdots,B_{V_{J-1}})$, and $T_V = \diag(T_{V_0},\cdots,T_{V_{J-1}})$. 
    Since $S_V$ is basically a block reordering of the parallel controlled applications of $S_{V_j}$, constructing of $S_V$ can be also achieved using $\mathcal{O}(1)$ queries. 
    Therefore, the query complexity of $S$ defined in~\cref{eqn:def_LCU_Cayley_transducer} is $\mathcal{O}(1)$. 

    Using the block representations 
    \begin{equation}
        S_V = \left( \begin{array}{cc}
            0 & E_V \\
            B_V & T_V
        \end{array} \right), \quad \mathcal{U}_p = \left( \begin{array}{ccc}
            \mathcal{J} & \mathcal{J}_{\perp} & 0 \\
            0 & 0 & I 
        \end{array}\right), \quad \Phi = \left( \begin{array}{ccc}
            -i I & & \\
             & i I & \\
             & & I
        \end{array} \right), 
    \end{equation}
    then we have 
    \begin{equation}\label{eqn:proof_Cayley_LCU_S_block}
        S = \left( \begin{array}{ccc}
            0 & 0 & -i \mathcal{J}^{\dagger} E_V \\
            0 & 0 & i \mathcal{J}_{\perp}^{\dagger} E_V\\
            B_V \mathcal{J} & B_V \mathcal{J}_{\perp} & T_V
        \end{array} \right) =: \left( \begin{array}{cc}
            0 & E \\
            B & T
        \end{array} \right), 
    \end{equation}
    where 
    \begin{equation}\label{eqn:proof_Cayley_LCU_def_EBT}
        E = \left( \begin{array}{cc}
            0 & -i \mathcal{J}^{\dagger} E_V
        \end{array} \right), \quad B = \left( \begin{array}{c}
            0 \\ 
            B_V \mathcal{J}
        \end{array} \right), \quad T = \left( \begin{array}{cc}
            0 & i \mathcal{J}_{\perp}^{\dagger} E_V\\
             B_V \mathcal{J}_{\perp} & T_V
        \end{array} \right). 
    \end{equation}

     Now let us verify that $S$ is the transducer of $\operatorname{Cay}(Y)$. 
     To this end, for any vector $\psi$, let us solve the equation 
     \begin{equation}
         S \left( \begin{array}{c}
              \psi \\
              c \\
              h
         \end{array} \right) = \left( \begin{array}{c}
              \widetilde{V} \psi \\
              c \\
              h
         \end{array} \right)
     \end{equation}
     where $\widetilde{V}$ is the unknown action matrix, and $c$ and $h$ are supposed to be the linear transformation of $\psi$. 
     Then we have 
     \begin{align}
         \widetilde{V}\psi &= -i \mathcal{J}^{\dagger} E_V h \\
         c &= i \mathcal{J}_{\perp}^{\dagger} E_V h \\
         h &=  B_V (\mathcal{J} \psi + \mathcal{J}_{\perp} c ) + T_V h. 
     \end{align}
     Let 
     \begin{equation}\label{eqn:proof_Cayley_LCU_y}
         y = \mathcal{J} \psi + \mathcal{J}_{\perp} c, 
     \end{equation}
     then we have 
     \begin{equation}\label{eqn:proof_Cayley_LCU_h}
         h = (I-T_V)^{-1} B_V y, 
     \end{equation}
     and thus 
     \begin{equation}\label{eqn:proof_Cayley_LCU_c}
         c = i \mathcal{J}_{\perp}^{\dagger} V y = i X^{\dagger} \psi + i Z c, \quad \widetilde{V}\psi = -i \mathcal{J}^{\dagger} V y = -i Y \psi -i X c,  
     \end{equation}
     where we use $E_V (I-T_V)^{-1} B_V = V$, which is a direct consequence of $E_{V_j} (I-T_{V_j})^{-1} B_{V_j} = V_j$ from~\cref{lem:composition_unitary}, and further denote 
     \begin{equation}
         X = \mathcal{J}^{\dagger} V \mathcal{J}_{\perp}, \quad Z = \mathcal{J}_{\perp}^{\dagger} V \mathcal{J}_{\perp}. 
     \end{equation}

     To further solve~\cref{eqn:proof_Cayley_LCU_h} and~\cref{eqn:proof_Cayley_LCU_c}, we need to note that 
     $(U_p\otimes I)^{\dagger} V (U_p\otimes I) = [Y, X; X^{\dagger},Z]$ is a Hermitian unitary and thus 
     \begin{equation}
         XX^{\dagger} = I - Y^2, \quad X^{\dagger} Y = -Z X^{\dagger}. 
     \end{equation}
     Then the first equation of~\cref{eqn:proof_Cayley_LCU_c} becomes 
     \begin{equation}\label{eqn:proof_Cayley_LCU_c_solu}
         c = i (I - iZ)^{-1} X^{\dagger} \psi = i X^{\dagger} (I + iY)^{-1} \psi, 
     \end{equation}
     where the last equation can be verified using $X^{\dagger} Y = -Z X^{\dagger}$. 
     Plugging~\cref{eqn:proof_Cayley_LCU_c_solu} back to the second equation in~\cref{eqn:proof_Cayley_LCU_c} gives 
     \begin{align}
         \widetilde{V}\psi & = -i Y \psi + X X^{\dagger} (I + iY)^{-1} \psi \\
         & = -i Y \psi + (I-Y^2) (I + iY)^{-1} \psi \\
         & = (-iY(I+iY) + I - Y^2 ) (I + iY)^{-1} \psi \\
         & = (I-iY)(I+iY)^{-1} \psi, 
     \end{align}
     so the action $\widetilde{V} = \operatorname{Cay}(Y)$. 
     Using~\cref{eqn:proof_Cayley_LCU_y},~\cref{eqn:proof_Cayley_LCU_h} and~\cref{eqn:proof_Cayley_LCU_c_solu}, we can also write 
     \begin{align}
         y &= \mathcal{J} \psi + i \mathcal{J}_{\perp} X^{\dagger} (I + iY)^{-1} \psi \\
         & = ( \mathcal{J}  + i (\mathcal{J} Y + \mathcal{J}_{\perp} X^{\dagger} ) ) (I + iY)^{-1} \psi \\
         & =  ( \mathcal{J}  + i (\mathcal{J} \mathcal{J}^{\dagger} V \mathcal{J}  + \mathcal{J}_{\perp} \mathcal{J}_{\perp}^{\dagger} V \mathcal{J} ) ) (I + iY)^{-1} \psi \\
         & = ( I  + i V ) \mathcal{J} (I + iY)^{-1} \psi, 
     \end{align}
     and thus 
     \begin{align}
         h = (I-T_V)^{-1} B_V ( I  + i V ) \mathcal{J} (I + iY)^{-1} \psi. 
     \end{align}
     Then the catalyst map is given as 
     \begin{equation}\label{eqn:proof_Cayley_LCU_Gamma}
         \Gamma \psi = \left( \begin{array}{c}
             c \\
             h
         \end{array} \right) = \left( \begin{array}{c}
             i X^{\dagger} (I + iY)^{-1}  \\
             (I-T_V)^{-1} B_V ( I  + i V ) \mathcal{J} (I + iY)^{-1} 
         \end{array} \right) \psi. 
     \end{equation}
     
     Finally, let us derive the the expression of $\Gamma^{\dagger}\Gamma$ and bound the related norms. 
     We start with $y$ defined in~\cref{eqn:proof_Cayley_LCU_y} and write 
     \begin{equation}
         y = \sum_{j=0}^{J-1} \ket{j} \otimes y_j. 
     \end{equation}
     Then each component $y_j$ is 
     \begin{align}
         y_j & = (\bra{j}\otimes I) y = (\bra{j}\otimes I) ( I  + i V ) \mathcal{J} (I + iY)^{-1} \psi  \\
         & = (\bra{j}\otimes I) ( I  + i V ) \left( \sum_{j'=0}^{J-1} \sqrt{q_{j'}} \ket{j'} \mathcal{J}_0 \right) (I + iY)^{-1} \psi \\
         & = \sqrt{q_j} (I + i V_j) \mathcal{J}_0 (I + iY)^{-1} \psi. 
     \end{align}
     Since $V_j$ is Hermitian unitary, we have $(I + i V_j)^{\dagger}(I + i V_j) = 2I$, and thus 
     \begin{equation}
         \|y_j\|^2 = 2 q_j \| (I + iY)^{-1} \psi \|^2, \quad \|y\|^2 = \sum \|y_j\|^2 = 2 \| (I + iY)^{-1} \psi \|^2. 
     \end{equation}
     Notice that we have $h = (I-T_V)^{-1} B_{V} y = \sum_{j=0}^{J-1} \ket{j}\otimes (I-T_{V_j})^{-1} B_{V_j} y_j $, 
     and, according to~\cref{lem:composition_unitary} that $\|(I-T_{V_j})^{-1} B_{V_j} y_j\|^2 = C_j \|y_j\|^2$, we have 
     \begin{equation}
         \|h\|^2 = \|(I-T_V)^{-1} B_{V} y\|^2 = \sum_{j=0}^{J-1} C_j \|y_j\|^2 = 2 \left(\sum_{j=0}^{J-1} q_j C_j\right) \| (I+iY)^{-1} \psi \|^2. 
     \end{equation}
     Also 
     \begin{equation}
         \| (I + iY)^{-1} \psi \|^2 = \psi^{\dagger} (I -iY^{\dagger})^{-1} (I + iY)^{-1} \psi = \psi^{\dagger} (I+Y^2)^{-1} \psi, 
     \end{equation}
     and then (using $y = \mathcal{J} \psi + \mathcal{J}_{\perp} c$ to be an orthogonal decomposition)
     \begin{equation}
         \|c\|^2 = \|y\|^2 - \|\psi\|^2 = 2 \| (I + iY)^{-1} \psi \|^2 -\|\psi\|^2 .
     \end{equation}
     Therefore 
     \begin{align}
         \psi^{\dagger} \Gamma^{\dagger}\Gamma \psi & = \|\Gamma \psi\|^2 = \|c\|^2 + \|h\|^2 = 2(1+\sum q_j C_j) \| (I + iY)^{-1} \psi\|^2 - \|\psi\|^2 \\
         &= \psi^{\dagger} \left( 2(1+\sum q_j C_j) (I+Y^2)^{-1} - I \right) \psi, 
     \end{align}
     which implies that 
     \begin{equation}
          \Gamma^{\dagger}\Gamma = 2(1+\sum q_j C_j) (I+Y^2)^{-1} - I, 
     \end{equation}
     and (using $(I+Y^2)^{-1} \preceq I $)
     \begin{equation}
         w(S) = \|\Gamma\|^2 = \sup_{\|\psi\|=1} \psi^{\dagger} \left( 2(1+\sum q_j C_j) (I+Y^2)^{-1} - I \right) \psi \leq 1 + 2 \sum q_j C_j. 
     \end{equation}

     To bound $K(S)$, from~\cref{eqn:proof_Cayley_LCU_def_EBT} we write 
     \begin{equation}
         I - T = \left( \begin{array}{cc}
            I & - i \mathcal{J}_{\perp}^{\dagger} E_V\\
             - B_V \mathcal{J}_{\perp} & I-T_V
        \end{array} \right). 
     \end{equation}
     To invert it, let us solve $(I-T)[x_0;x_1] = [z_0;z_1]$, i.e., 
     \begin{align}
         x_0 - i \mathcal{J}_{\perp}^{\dagger} E_V x_1 &= z_0 \\
         - B_V \mathcal{J}_{\perp} x_0 + (I-T_V) x_1 &= z_1. 
     \end{align}
     Pluging the second equation back to the first equation gives 
     \begin{equation}
         x_0 - i \mathcal{J}_{\perp}^{\dagger} E_V (I-T_V)^{-1}( B_V \mathcal{J}_{\perp} x_0 + z_1 ) = z_0, 
     \end{equation}
     which, using $E_V (I-T_V)^{-1} B_V = V$, becomes 
     \begin{equation}
         x_0 = (I - iZ)^{-1} z_0 + i (I - iZ)^{-1} \mathcal{J}_{\perp}^{\dagger} E_V (I-T_V)^{-1} z_1. 
     \end{equation}
     The $x_1$ component becomes 
     \begin{align}
         x_1 &= (I-T_V)^{-1} B_V \mathcal{J}_{\perp} x_0 + (I-T_V)^{-1}z_1 \\
         &= (I-T_V)^{-1} B_V \mathcal{J}_{\perp}(I - iZ)^{-1} z_0 \\
         &\quad + i (I-T_V)^{-1} B_V \mathcal{J}_{\perp} (I - iZ)^{-1} \mathcal{J}_{\perp}^{\dagger} E_V (I-T_V)^{-1} z_1 + (I-T_V)^{-1}z_1. 
     \end{align}
     The expressions of $x_0$ and $x_1$ means that $(I-T)^{-1}$ equals
     {\small
     \begin{equation}
         \left( \begin{array}{cc}
             (I - iZ)^{-1} & i (I - iZ)^{-1} \mathcal{J}_{\perp}^{\dagger} E_V (I-T_V)^{-1} \\
             (I-T_V)^{-1} B_V \mathcal{J}_{\perp}(I - iZ)^{-1} & (I-T_V)^{-1} + i (I-T_V)^{-1} B_V \mathcal{J}_{\perp} (I - iZ)^{-1} \mathcal{J}_{\perp}^{\dagger} E_V (I-T_V)^{-1}
         \end{array} \right). 
     \end{equation}
     }
     Since $((I-iZ)^{-1})^{\dagger}(I-iZ)^{-1} = (I+Z^2)^{-1} \preceq I$ and $\|(I-iZ)^{-1}\| \leq 1$, and according to~\cref{lem:composition_unitary}, 
     \begin{align}
         \|(I-T_V)^{-1}B_V\| &\leq \max \|(I-T_{V_j})^{-1}B_{V_j}\| = \max \sqrt{C_j}, \\
         \|E_V(I-T_V)^{-1}\| &\leq \max \| E_{V_j}(I-T_{V_j})^{-1} \| = \max \sqrt{C_j}, \\
         \|(I-T_V)^{-1}\| &\leq \max \|(I-T_{V_j})^{-1}\| \leq \max C_j, 
     \end{align}
     the spectral norms of the four blocks of $(I-T)^{-1}$ are bounded by $1$, $\max \sqrt{C_j}$, $\max \sqrt{C_j}$, and $2 \max C_j$, respectively. 
     Then 
     \begin{equation}
         K(S) = \|(I-T)^{-1}\| \leq 1 + 2 \max \sqrt{ C_j} + 2 \max C_j \leq 6 \max C_j. 
     \end{equation}
\end{proof}

\subsection{Algorithm description}\label{sec:algorithm_description}

We first introduce or recall a few notations: 
\begin{equation}
    \widehat{A}_j = \frac{A_j}{\alpha_j}, \quad \lambda = \sum_{j=0}^{J-1} c_j \alpha_j, \quad p_j = \frac{c_j \alpha_j}{\lambda}, \quad \overline{C} = \sum_{j=0}^{J-1} p_j C_j, \quad C_{\max} = \max_j C_j. 
\end{equation}
Consider the linear combination 
\begin{equation}
    A = \sum_{j=0}^{J-1} c_j A_j, \quad \|A\| \leq a
\end{equation}
where $a$ is a known upper bound of $\|A\|$, and we assume $a \leq \lambda$ as $\lambda$ is the naive upper bound. 
The full transducer-based LCU algorithm is described as follows. 

\paragraph{Step 1: Construct the Cayley LCU transducer.} In this step we will construct a transducer of
\begin{equation}
    \operatorname{Cay}\left(\frac{A}{4\lambda}\right), \quad \text{where}\quad 
    \frac{A}{4\lambda} = \sum_{j=0}^{J-1} \frac{p_j}{4} \widehat{A}_j. 
\end{equation}
Notice that the coefficients $q_j = p_j/4$ only sum to $1/4$ rather than $1$, so we append $q_J = 3/4$ and $V_J$ be any Hermitian block-encoding of a zero matrix. 
Then we apply the algorithm described in~\cref{sec:transducer_LCU_normalized}, which gives a transducer of $\operatorname{Cay}(\sum_{j=0}^{J-1} q_j \widehat{A}_j) = \operatorname{Cay}(A/(4\lambda))$ according to~\cref{thm:Cayley_LCU}.

\paragraph{Step 2: Perform inverse Cayley transformation and reduce the normalization factor.}
Let $Y = A/(4\lambda)$ and $V = \operatorname{Cay}(Y)$. 
The main idea of this step is to write down a circuit involving $V$ to construct a block-encoding of $A/(16a)$, and use~\cref{cor:composition_general} to construct a transducer of this circuit rather than directly implementing it. 

We first write down the circuit expression, which is a QSVT circuit on a small ordinary LCU. 
Specifically, the standard LCU can block encode 
\begin{equation}
    Z = \frac{V^\dagger - V}{2i} = 2Y (I+Y^2)^{-1}. 
\end{equation}
Notice that $\|Y\| \leq a/(4\lambda)$, we choose a polynomial $P(x)$ such that 
\begin{equation}
    P\left(\frac{2b}{1+b^2}\right) \approx \frac{b}{16(a/(4\lambda))}, \quad |b| \leq \frac{a}{4\lambda}. 
\end{equation}
Such a polynomial simultaneously inverts the Cayley transformation and reduces the normalization factor of $A$ from $\mathcal{O}(\lambda)$ to $\mathcal{O}(a)$.
A specific construction of this polynomial is given in~\cref{sec:polynomial_inv_Cayley_amplification}. 
Then the QSVT circuit with block-encoding of $Z$ yields a block-encoding of $\frac{Y}{16(a/(4\lambda))} = A/(16a)$. 

For the actual implementation, since we have constructed a transducer of $V$ in step 1, we can apply~\cref{cor:composition_general} to turn the pre-described circuit into a transducer of the full QSVT unitary $U_P$, where $\left(\left\langle 0^r\right| \otimes I\right) U_P\left(\left|0^r\right\rangle \otimes I\right)=\tilde{A} \approx \frac{A}{16 a}$.

\paragraph{Step 3: Convert the transducer to a block-encoding via finite reuse.}
In step 2, we have constructed a transducer of $U_P$. 
Now we directly apply~\cref{thm:finite-reuse-high-order} to construct a block-encoding of approximately $A/(16a)$.

\subsection{Complexity analysis}

\begin{theorem}[Complexity estimate of the transducer-based LCU]\label{thm:complexity_main}
    Let $U_j$ be a Hermitian block-encoding of $A_j/\alpha_j$ with normalization factor $\alpha_j$, and suppose that $U_j$ satisfies~\cref {assump:U} with query complexity $C_j$. 
    Let $A = \sum_{j=0}^{J-1} c_j A_j$ with coefficients $c_j > 0$, and suppose that an upper bound $a \geq \|A\|$ is known. 
    Then, for any $\epsilon > 0$, the algorithm described in~\cref{sec:algorithm_description} gives a block-encoding of $A'$ such that $\|A'-A/(16a)\| \leq \mathcal{O}(\epsilon)$, using 
    \begin{equation}
        \mathcal{O}\left( C_{\max} \log\left( \frac{1}{\epsilon} \right) + \frac{\lambda}{a} \overline{C} \left(\log\left(\frac{1}{\epsilon}\right)\right)^2 \right)
    \end{equation}
    queries to the primitive oracles, where 
    \begin{equation}
        \lambda = \sum_{j=0}^{J-1} c_j \alpha_j, \quad C_{\max} = \max_j C_j, \quad \overline{C} = \frac{\sum_{j=0}^{J-1} c_j \alpha_j C_j }{\lambda}. 
    \end{equation}
\end{theorem}

\begin{proof}

    The proof can be completed by walking through the algorithm and combining all the previously established complexity estimates for intermediate subroutine. 
    
    In Step 1, we construct a transducer $S_{\operatorname{Cay}}$ of $\operatorname{Cay}(A/(4\lambda))$. 
    \Cref{thm:Cayley_LCU} tells that 
    \begin{equation}
        w(S_{\operatorname{Cay}}) \leq 1 + 2 \sum_{j=0}^{J-1} \frac{c_j \alpha_j }{4 \lambda} C_j = \mathcal{O}(\overline{C}), \quad K(S_{\operatorname{Cay}}) \leq 6 \max_j C_j = \mathcal{O}(C_{\max}), 
    \end{equation}
    and the construction of $S_{\operatorname{Cay}}$ needs $\mathcal{O}(1)$ query complexity. 
    
    The QSVT circuit in Step 2 has query length to be $\mathcal{O}(\operatorname{deg}(P(x)))$, according to~\cref{lem:QSVT}, where $P(x)$ is the polynomial for inverse Cayley transformation with amplification. 
    By~\cref{lem:polynomial_inv_Cayley_amplification}, for any $0 < \delta_1 < 1/2$ to be determined later (and choosing $\rho = a/(4\lambda)$ in~\cref{lem:polynomial_inv_Cayley_amplification}), such a QSVT circuit gives a block-encoding of $\widetilde{A}$ such that 
    \begin{equation}
        \left\| \widetilde{A} - \frac{A}{16a} \right\| \leq \delta_1, 
    \end{equation}
    with 
    \begin{equation}
        \operatorname{deg}(P(x)) = \mathcal{O}\left( \frac{\lambda}{a} \log\left( \frac{1}{\delta_1} \right) \right). 
    \end{equation}
    Using~\cref{cor:composition_general}, we can construct a transducer $S$ of a block-encoding of $\widetilde{A}$ with 
    \begin{equation}
        K(S) \leq \mathcal{O}\left( K(S_{\operatorname{Cay}}) + \operatorname{deg}(P(x)) w(S_{\operatorname{Cay}}) \right) = \mathcal{O}\left( C_{\max} + \frac{\lambda}{a} \overline{C} \log\left(\frac{1}{\delta_1}\right) \right), 
    \end{equation}
    and the query complexity of $S$ is still $\mathcal{O}(1)$. 

    The last step implements a finite reuse of $S$. 
    For any $\delta_2 > 0$, according to~\cref{thm:finite-reuse-high-order}, we can implement a unitary $U$, which is a block-encoding of $\widetilde{\widetilde{A}}$ such that 
    \begin{equation}
        \| \widetilde{\widetilde{A}} - \widetilde{A}\| \leq \mathcal{O}(\delta_2), 
    \end{equation}
    using 
    \begin{equation}\label{eqn:proof_complexity_query_complexity}
        \mathcal{O}\left( K(S) \log\left( \frac{1}{\delta_2} \right) \right) = \mathcal{O}\left( C_{\max} \log\left( \frac{1}{\delta_2} \right) + \frac{\lambda}{a} \overline{C} \log\left(\frac{1}{\delta_1}\right) \log\left( \frac{1}{\delta_2} \right)\right)
    \end{equation}
    queries to $S$, which turns to asymptotically the same number of queries to the primitive oracles as each $S$ requires $\mathcal{O}(1)$ queries. 
    As $ \| \widetilde{\widetilde{A}} - \widetilde{A}\| \leq \mathcal{O}(\delta_2)$ and $\left\| \widetilde{A} - \frac{A}{16a} \right\| \leq \delta_1$, we can choose 
    \begin{equation}
        \delta_1 = \Theta(\epsilon), \quad \delta_2 = \Theta(\epsilon)
    \end{equation}
    to make sure that this unitary $U$ is a block-encoding of $A/(16a)$ to error at most $\mathcal{O}(\epsilon)$. 
    Plugging the choice of $\delta_1$ and $\delta_2$ back to~\cref{eqn:proof_complexity_query_complexity} completes the proof. 
\end{proof}

\section{Applications to sparse matrix functions}\label{sec:applications}

In this section, we discuss applications of our transducer-based LCU algorithm for implementing sparse matrix functions, including the construction of sparse matrix block-encoding, as well as solving sparse Hamiltonian simulation and linear systems of equations problems.

Let $A \in \mathbb{C}^{N\times N}$ be a Hermitian matrix with at most $d$ non-zero entries per row or column. 
We assume access to its standard black-box unitary oracles which compute the positions and values of non-zero entries. 
Specifically, 
\begin{align}
	O_F\ket{i,\ell} &= \ket{i,f(i,\ell)},\qquad i \in [N], \ell \in [d]\label{eq:OF}\\
	O_A\ket{i,k,z} &= \ket{i,k,z\oplus A_{ik}}, \qquad i,k\in [N],\label{eq:OA}
\end{align}
where $f(i,\ell)$ is the column index of the $\ell$-th non-zero entry in the $i$-th row, $z \in \{0,1\}^b$ is a $b$-bit string, and $A_{ik}$ denotes the matrix entry at $i$-th row and $k$-th column. 
These two oracles $O_F$ and $O_A$ will serve as the ``primitive oracles'' in~\cref{assump:U}. 
We denote $\|A\|$ as the spectral norm of $A$, $\|A\|_{\max} = \max_{i,k} |A_{ik}|$, and $\norm{A}_{1\to 2}:=\max_k\left(\sum_i |A_{ik}|^2\right)^{1/2}$.

\subsection{Sparse matrix block-encoding}\label{sec:applications_sparse_block_encoding}

Now we discuss how to efficiently construct a block-encoding of a sparse matrix. 
The approach is based on a decomposition $A = \sum_{j=0}^{m} A_j$ modified from~\cite{Low2019}, which splits the original sparse matrix into several according to the entries' magnitudes.
Then we can construct the block-encoding of each $A_j$ using the standard sparse matrix block-encoding with amplitude amplification in~\cite{Low2019}, and apply our transducer-based LCU to construct $A$. 

\subsubsection{A review of the standard sparse block-encoding with amplification}

We first state the standard approach for block encoding a sparse matrix $H$ with amplitude amplification in~\cite{Low2019}. 
The algorithm first constructs two sets of mutually orthogonal quantum states $\{\chi_k\},\,\{\bar{\chi}_j\}$ such that their overlap $\langle \bar{\chi}_j | \chi_k \rangle = H_{jk}/(d\|H\|_{\max})$. 
Then $H$ is block-encoded by the product of controlled-state preparation unitary $U = U_{\text{row}}^{\dagger} U_{\text{col}}$, where $U_{\text{col}}|k\rangle_s|0\rangle_a = |\chi_k\rangle, U_{\text{row}}|j\rangle_s|0\rangle_a = |\bar{\chi}_j\rangle$. 
Both $U_{\text{row}}^{\dagger} $ and $U_{\text{col}}$ can be constructed using $\mathcal{O}(1)$ queries. 
The normalization factor $d\|H\|_{\max}$ is induced by the normalization factor $\sqrt{\|H\|_{\max}}$ in $U_{\text{row}}^{\dagger} $ and $U_{\text{col}}$, due to the controlled rotation in constructing them, which require each non-zero entry to be shrinked smaller than 1. 
However, such a normalization factor is not optimal as the largest possible normalization factor for $U_{\text{row}}^{\dagger} $ and $U_{\text{col}}$ is $\sqrt{\|H\|_1}$, so the algorithm applies the QSVT-based uniform singular value amplification technique to boost the normalization factor of $U_{\text{row}}^{\dagger} $ and $U_{\text{col}}$ to $\sqrt{\Lambda_1}$ for any $\Lambda_1 \geq \|H\|_1$. 
This process requires $\mathcal{O}( \sqrt{d\|H\|_{\max} / \Lambda_1} \log(1/\epsilon) )$ queries, leading to a block-encoding of approximately $H/\Lambda_1$. 

We formally summarize the complexity of this procedure in the next lemma, which is from Theorem~10 of \cite{Low2019}. 

\begin{lemma}[Sparse block-encoding by amplitude multiplication]\label{lem:sparse-block-encoding}
Let $H\in C^{N\times N}$ be Hermitian and $d$-sparse. 
Then, for any $0 < \delta < 1$ and $\Lambda_1 \geq \|H\|_1$, we can construct a block-encoding of a Hermitian matrix $\tilde{H}/\alpha$ with $\alpha = \Theta(\Lambda_1)$ and $\|\tilde{H} - H\| = \mathcal{O}(\|H\| \delta)$, using
\begin{equation}
\mathcal{O} \left(\sqrt{\frac{d \|H\|_{\max}}{\Lambda_1}} \log \left(\frac{1}{\delta}\right)\right)  
\end{equation}
queries to the sparse matrix oracles. 
\end{lemma}

\subsubsection{Matrix decomposition}

Now we decompose the matrix $A$ according to the following rule. 
Let 
\begin{equation}
    \mu_0 = \frac{\|A\|_{1\to 2}}{\sqrt{d}}, \quad \mu_{j} = 2^j \mu_0 = \frac{2^j \|A\|_{1\to 2}}{\sqrt{d}}, 
\end{equation}
and $m$ be the smallest index such that $\mu_{m} \geq \|A\|_{\max}$. 
It is clear that 
\begin{equation}
    m = \mathcal{O}(\log(\sqrt{d} \|A\|_{\max} / \|A\|_{1\to 2} )). 
\end{equation} 
Define the matrix $A_j$ to be 
\begin{align}
    (A_0)_{ik} &= A_{ik} \mathbf{1}\{ |A_{ik}| \leq \mu_0 \}, \label{eqn:sparse_decomposition_def_1}\\
    (A_j)_{ik} &= A_{ik} \mathbf{1}\{ \mu_{j-1} < |A_{ik}| \leq \mu_j \}, \quad 1 \leq j \leq m, \label{eqn:sparse_decomposition_def_2}
\end{align}
i.e., $A_j$ consists of $A$'s non-zero entries between $\mu_{j-1}$ and $\mu_j$. 
Thus $A = \sum_{j=0}^{m} A_j$.

To apply our transducer-based LCU, we estimate relative quantities in the next result. 

\begin{lemma}\label{lem:sparse_decomposition_norm}
    Let $A = \sum_{j=0}^{m} A_j$ where each $A_j$ is defined through~\cref{eqn:sparse_decomposition_def_1,eqn:sparse_decomposition_def_2}. 
    Then, for any $0 < \epsilon < 1$, we can implement block-encodings $U_j$ of $\widetilde{A}_j/\alpha_j$ with query complexity $C_j$, such that 
    \begin{align}
        \lambda &= \sum_{j=0}^{m} \alpha_j = \mathcal{O}\left( \sqrt{d}  \|A\|_{1\to 2} \right), \\
        \overline{C} &= \frac{\sum_{j=0}^{m} \alpha_j C_j }{\lambda} = \mathcal{O}\left(  \log\left( \frac{\sqrt{d}\|A\|_{\max}}{\|A\|_{1\to 2}} \right) \log\left(\frac{d \|A\|_{1\to 2}}{\epsilon}\right) \right), \\
        C_{\max} &= \max_j C_j = \mathcal{O}\left( \frac{\sqrt{d}\|A\|_{\max} }{\|A\|_{1\to 2}} \log\left(\frac{d \|A\|_{1\to 2}}{\epsilon}\right)\right), 
    \end{align}
    and 
    \begin{equation}
        \|\widetilde{A} - A\| \leq \mathcal{O}( \epsilon), \quad \widetilde{A} := \sum_{j=0}^{m} \widetilde{A}_j. 
    \end{equation}
    
\end{lemma}

\begin{proof}
    We first show that 
    \begin{equation}\label{eqn:proof_sparse_decomposition_eq1}
        \|A_0\|_1 \leq \sqrt{d} \|A\|_{1\to 2}, \quad \|A_j\|_1 \leq \frac{\|A\|_{1\to 2}^2 }{\mu_{j-1}},\quad \forall 1 \leq j \leq m. 
    \end{equation}
    For $A_0$, as all the entries have absolute values smaller than $\mu_0$, and there are at most $d$ non-zero entries per column, we have 
    \begin{equation}
        \|A_0\|_1 \leq d \mu_0 = \sqrt{d} \|A\|_{1\to 2}. 
    \end{equation}
    For $j \geq 1$, for every non-zero entry $(A_j)_{ik}$, we first use $|(A_j)_{ik}| > \mu_{j-1}$ to obtain $|(A_j)_{ik}| \leq \frac{|(A_j)_{ik}|^2}{\mu_{j-1}}$, and then 
    \begin{equation}
        \|A_j\|_{1} = \max_i \sum_{k} |(A_j)_{ik}| \leq \frac{\max_i \sum_{k} |(A_j)_{ik}|^2 }{\mu_{j-1}} \leq \frac{\|A\|_{1\to 2}^2}{\mu_{j-1}}. 
    \end{equation}
    We have proved~\cref{eqn:proof_sparse_decomposition_eq1}. 

    According to~\cref{eqn:proof_sparse_decomposition_eq1}, we can choose the normalization factor of $\widetilde{A}_j$ as 
    \begin{equation}
        \alpha_0 = \sqrt{d} \|A\|_{1 \to 2}, \quad \alpha_j = \frac{\|A\|_{1\to 2}^2}{\mu_{j-1}} = \frac{\sqrt{d} \|A\|_{1\to 2}}{2^{j-1}}, \,\, \forall 1 \leq j \leq m. 
    \end{equation}
    Taking the sum, we first have 
    \begin{equation}
        \lambda = \sum_{j=0}^{m} \alpha_j = \Theta\left( \sqrt{d}  \|A\|_{1\to 2} \right). 
    \end{equation}
    For any $\delta>0$ to be specified later, using~\cref{lem:sparse-block-encoding}, we can construct block-encodings $U_j$ of $\widetilde{A}_j/\alpha_j$ with query complexity $C_j$, such that 
    \begin{equation}
        \|\widetilde{A}_j - A_j\| \leq \mathcal{O}(\|A_j\|\delta ), \quad C_j = \mathcal{O}\left(\sqrt{\frac{d\mu_j}{\alpha_j}} \log\left(\frac{1}{\delta}\right)\right). 
    \end{equation}
    For error, we simply take the sum with respect to $j$ and obtain 
    \begin{equation}
        \| \widetilde{A} - A \| \leq \sum_{j=0}^{m} \|\widetilde{A}_j - A_j\| = \mathcal{O}(\lambda \delta). 
    \end{equation}
    So we can choose $\delta = \epsilon/\lambda$ to bound the overall error by $\mathcal{O}(\epsilon)$, and the corresponding query complexity becomes 
    \begin{equation}
        C_j = \mathcal{O}\left(\sqrt{\frac{d\mu_j}{\alpha_j}} \log\left(\frac{\lambda}{\epsilon}\right)\right) = \mathcal{O}\left(\sqrt{\frac{d\mu_j}{\alpha_j}} \log\left(\frac{d \|A\|_{1\to 2}}{\epsilon}\right)\right). 
    \end{equation}
    Then, using $\mu_j = 2 \mu_{j-1} = \frac{2\|A\|^2_{1\to 2}}{\alpha_j}$ for $j \geq 1$ and $\mu_0 = \frac{\|A\|_{1\to 2}^2}{\alpha_0}$, we can compute 
    \begin{align}
        \overline{C} &= \frac{\sum_{j=0}^{m} \alpha_j C_j }{\lambda} = \mathcal{O} \left( \frac{1}{\lambda} \sum_{j=0}^{m} \alpha_j \sqrt{\frac{d \mu_j}{{\alpha_j}} }\log\left(\frac{d \|A\|_{1\to 2}}{\epsilon}\right) \right) \\
        & = \mathcal{O}\left( m \log\left(\frac{d \|A\|_{1\to 2}}{\epsilon}\right) \right) = \mathcal{O}\left(  \log\left( \frac{\sqrt{d}\|A\|_{\max}}{\|A\|_{1\to 2}} \right) \log\left(\frac{d \|A\|_{1\to 2}}{\epsilon}\right) \right). 
    \end{align}

    To bound the maximum complexity, notice that $\sqrt{d\mu_0/\alpha_0} = \mathcal{O}(1)$ and $\sqrt{d\mu_j/\alpha_j} = \mathcal{O}(2^j)$, we have 
    \begin{equation}
        C_{\max} = \max_j C_j = \mathcal{O}\left(  2^{m-1} \log\left(\frac{d \|A\|_{1\to 2}}{\epsilon}\right) \right) = \mathcal{O}\left( \frac{\sqrt{d}\|A\|_{\max} }{\|A\|_{1\to 2}} \log\left(\frac{d \|A\|_{1\to 2}}{\epsilon}\right)\right), 
    \end{equation}
    where the last estimate is from $\mu_{m-1} < \|A\|_{\max}$ according to the definition of $m$. 
\end{proof}

\subsubsection{Complexity}

The complexity of applying transducer-based LCU to implement the above decomposition directly follows from~\cref{thm:complexity_main} and~\cref{lem:sparse_decomposition_norm}. 

\begin{theorem}[Sparse block-encoding via transducer-based LCU]\label{thm:sparse_block_encoding_transducer}
    Let $A$ be a $d$-sparse Hermitian matrix with a known upper bound $a \geq \|A\|$ and known $\|A\|_{\max}$, $\|A\|_{1\to 2}$. 
    Then, for any $0 < \epsilon \leq \mathcal{O}(a)$, we can construct a block-encoding of $A'/\alpha$ with $\alpha = \Theta(a)$ and $\|A' - A\| \leq \mathcal{O}(\epsilon)$, using 
    \begin{equation}
        \mathcal{O}\left( \left(\frac{\sqrt{d}\|A\|_{\max} }{\|A\|_{1\to 2}} + \frac{\sqrt{d}\|A\|_{1\to 2}}{a} \log\left( \frac{\sqrt{d}\|A\|_{\max}}{\|A\|_{1\to 2}} \right)   \log\left(\frac{a}{\epsilon}\right)\right) \log\left(\frac{d \|A\|_{1\to 2}}{\epsilon}\right) \log\left( \frac{a}{\epsilon} \right) \right)
    \end{equation}
    queries to the sparse oracle assess to $A$. 
\end{theorem}

\begin{proof}
    By~\cref{lem:sparse_decomposition_norm}, we have block-encodings $U_j$ of $\widetilde{A_j}/\alpha_j$ such that $\|\widetilde{A} - A\| \leq \mathcal{O}(\epsilon)$, $\widetilde{A} = \sum_{j=0}^{m} \widetilde{A}_j$. 
    Since $\|A\| \leq a$, we can use $\beta = \Theta(a)$ as an upper bound of $\|\widetilde{A}\|$. 
    Then,~\cref{thm:complexity_main} gives a block-encoding of $A'/(16\beta)$ such that $\|A'/(16\beta) - \widetilde{A}/(16\beta)\| \leq \mathcal{O}(\epsilon/\beta)$, which in turn implies $\|A' - A\| \leq \mathcal{O}(\epsilon)$. 

    According to~\cref{thm:complexity_main} and~\cref{lem:sparse_decomposition_norm}, query complexity of this algorithm is 
    \begin{align}
        & \quad \mathcal{O}\left( C_{\max} \log\left( \frac{a}{\epsilon} \right) + \frac{\lambda}{a} \overline{C} \left(\log\left(\frac{a}{\epsilon}\right)\right)^2 \right) \\
        & = \mathcal{O}\left( \frac{\sqrt{d}\|A\|_{\max} }{\|A\|_{1\to 2}} \log\left(\frac{d \|A\|_{1\to 2}}{\epsilon}\right) \log\left( \frac{a}{\epsilon} \right) \right.\\
        &\qquad\quad + \left.\frac{\sqrt{d}\|A\|_{1\to 2}}{a} \log\left( \frac{\sqrt{d}\|A\|_{\max}}{\|A\|_{1\to 2}} \right) \log\left(\frac{d \|A\|_{1\to 2}}{\epsilon}\right)  \left(\log\left(\frac{a}{\epsilon}\right)\right)^2 \right). 
    \end{align}
\end{proof}

A particular useful case is when the matrix $A$ has spectral norm bounded by $1$. 
Then we may choose $a = 1$, and use the inequalities $\|A\|_{\max}/\|A\|_{1\to 2} \leq 1$ and $\|A\|_{1\to 2} \leq \|A\| \leq 1$ to simplify the expression of the query complexity. 

\begin{corollary}[Sparse block-encoding via transducer-based LCU for normalized matrix]\label{cor:sparse_block_encoding_normalized}
    Let $A$ be a $d$-sparse Hermitian matrix with $\|A\| \leq 1$. 
    Then, for any $0 < \epsilon \leq \mathcal{O}(1)$, we can construct a block-encoding of $A'/\alpha$ with $\alpha = \mathcal{O}(1)$ and $\|A' - A\| \leq \mathcal{O}(\epsilon)$, using 
    \begin{equation}
        \mathcal{O}\left( \sqrt{d} \log\left( d \right) \log\left(\frac{d}{\epsilon}\right)  \left(\log\left(\frac{1}{\epsilon}\right)\right)^2 \right)
    \end{equation}
    queries to the sparse oracle assess to $A$. 
\end{corollary}

\subsection{Sparse Hamiltonian simulation}\label{sec:applications_sparse_Ham_Sim}

For an evolution time $T > 0$ and a target error $\epsilon > 0$, the goal of the Hamiltonian simulation problem for a Hamiltonian $A$ is to approximate $e^{-i A T}$ to error $\epsilon$. 
We assume $\|A\| \leq 1$ as we can rescale the dynamics to absorb the spectral norm into the evolution time. 

QSVT algorithm can solve the Hamiltonian simulation problem with $\mathcal{O}(\alpha T +\log(1/\epsilon))$ queries to the block-encoding of $A/\alpha$. 
We can directly apply our~\cref{cor:sparse_block_encoding_normalized} to obtain such a block-encoding. 
A nuance is that our block-encoding is not exact. 
By $\|e^{-i A T} - e^{-i A' T}\| \leq T\|A-A'\|$, it suffices to choose the block-encoding error $\mathcal{O}(\epsilon/T)$ to bound the overall simulation error still by $\mathcal{O}(\epsilon)$. 
Summarizing these discussions, we have the following result. 

\begin{theorem}\label{thm:sparse_Ham_sim}
    Let $A$ be a $d$-sparse Hamiltonian with $\|A\| \leq 1$. 
    Then, for an evolution time $T$ and a target error $\epsilon$, there exists a quantum algorithm which implements a block-encoding of $V$ such that $\|V - e^{-i A T}\| \leq \mathcal{O}(\epsilon)$, using 
    \begin{equation}
        \mathcal{O}\left( \left(T+\log\left(\frac{1}{\epsilon}\right)\right) \sqrt{d} \log\left( d \right) \log\left(\frac{dT}{\epsilon}\right)  \left(\log\left(\frac{T}{\epsilon}\right)\right)^2 \right)
    \end{equation}
    queries to the sparse oracle access of $A$. 
\end{theorem}

\subsection{Sparse linear system algorithm}\label{sec:applications_sparse_linear_system}

The quantum linear system problem considers 
\begin{equation}
    A x = \ket{b}, 
\end{equation}
and the goal of a quantum linear system algorithm is to prepare a quantum state $\ket{\widetilde{x}}$ such that 
\begin{equation}
    \| \ket{\widetilde{x}} - \ket{x}\| \leq \epsilon, 
\end{equation}
where $\ket{x} = x/\|x\|$ is the normalized solution of the linear system. 
Since the goal is the normalized solution, we assume the right hand side vector $\ket{b}$ to be normalized without loss of generality. 
For the coefficient matrix, we also assume $\|A\| = 1$ and denote its condition number by $\kappa = \|A\|\|A^{-1}\|$. 
We further restrict $A$ to be a Hermitian matrix, and general non-Hermitian coefficient matrix can be reduced to this Hermitian case by the standard dilation trick with the help of one extra qubit~\cite{HarrowHassidimLloyd2009}. 
Query complexity of a quantum linear system algorithm is counted as the number of queries to $O_F$ and $O_A$, and the state preparation oracle $O_b: \ket{0} \mapsto \ket{b}$. 

There have been several existing works which solve the quantum linear system problem using the matrix block-encoding input model. 
Our sparse matrix block-encoding technique can be naturally combined with those algorithms to improve the dependence on the sparsity. 

For example, with the block-encoding of $A'/\alpha$, the discrete adiabatic algorithm~\cite{CostaAnSandersEtAl2022} can solve the linear system $A'x' = \ket{b}$ to error in the normalized state at most $\epsilon$, using $\mathcal{O}\left( \alpha \|A'^{-1}\| \log(1/\epsilon) \right)$ queries to $O_F,O_A$ and $O_b$. 
Furthermore, for any $\epsilon>0$, we can ensure $\|x/\|x\| - x'/\|x'\|\| \leq \mathcal{O}(\epsilon)$, by choosing the error in the block-encoding to be \footnote{This claim can be proved as follows. 
Let $\|A'-A\| =\delta$ and write $A' = A+E$, $u = x/\|x\|$, $v = x'/\|x'\|$, and $r = \|x\|/\|x'\|$. 
Since $x = x' + A^{-1}Ex'$ and $\|ru-v\| \leq \delta \|A^{-1}\|$ (we force $\delta \|A^{-1}\| \leq 1/2$), we can bound $\operatorname{Re}\braket{u|v} \geq \frac{r^2+1-(\delta \|A^{-1}\|)^2}{2r} \geq \sqrt{1-(\delta \|A^{-1}\|)^2}$. Then the claimed bound follows from $\|u-v\| = \sqrt{2 - 2\operatorname{Re}\braket{u|v}} = \mathcal{O}(\delta\|A^{-1}\|)$. } 
$\mathcal{O}(\epsilon/\|A^{-1}\|)$. 
Putting this with~\cref{cor:sparse_block_encoding_normalized} together, the query complexity becomes $\mathcal{O}(\sqrt{d} \kappa \operatorname{polylog}(d\kappa/\epsilon) ) $.

Alternatively, we may apply the tunable variable time amplitude amplification (tunable VTAA) \cite{LowSu2024quantumlinearalgorithmoptimal} to trade off between matrix query complexity and state preparation cost. 
For $A'x' = b$, it requires $\mathcal{O}(\alpha \kappa \log(1/p') \log(\log(1/p')/\epsilon) )$ queries to the block-encoding of $A'/\alpha$, and $\mathcal{O}(1/\sqrt{p'})$ to $O_b$, where $1/\sqrt{p'} = \|A'^{-1}\|/\|A'^{-1}\ket{b}\|$. 
Using $\|A'^{-1}\|/\|A'^{-1}\ket{b}\| \sim \|A^{-1}\|/\|A^{-1}\ket{b}\|$, $\|A'^{-1}\| \sim \|A^{-1}\|$, and choosing a similar block-encoding accuracy as in the discrete adiabatic approach, we can also solve $Ax = \ket{b}$ with error at most $\mathcal{O}(\epsilon)$. 

We summarize these two approaches in the following result. 

\begin{theorem}\label{thm:sparse_linear_system}
    Let $A$ be a $d$-sparse Hermitian matrix with $\|A\| = 1$ and condition number $\kappa$, and $\ket{b}$ be a normalized vector. 
    Then, for any $\epsilon \leq \mathcal{O}(1/\kappa)$: 
    \begin{enumerate}
        \item The discrete adiabatic algorithm can prepare a quantum state $\ket{\widetilde{x}}$ such that
        $\|\ket{\widetilde{x}} - \frac{A^{-1}\ket{b}}{\|A^{-1}\ket{b}\|}\| \leq \mathcal{O}(\epsilon)$, 
        using 
    \begin{equation}
        \mathcal{O}\left( \kappa \sqrt{d} \log\left( d \right) \log\left(\frac{d\kappa}{\epsilon}\right)  \left(\log\left(\frac{\kappa}{\epsilon}\right)\right)^2 \log\left(\frac{1}{\epsilon}\right) \right)
    \end{equation}
    queries to the sparse oracle access to $A$ and the state preparation oracle of $\ket{b}$.
    \item Alternatively, the tunable VTAA approach can achieve the same goal, using 
    \begin{equation}
        \mathcal{O}\left( \kappa \sqrt{d} \log\left( d \right) \log\left(\frac{d\kappa}{\epsilon}\right)  \left(\log\left(\frac{\kappa}{\epsilon}\right)\right)^2 \log\left(\frac{1}{p} \right) \log\left( \frac{\log(1/p)}{\epsilon}\right) \right)
    \end{equation}
    queries to the sparse oracle access to $A$, and 
    \begin{equation}
        \mathcal{O}\left( \frac{1}{\sqrt{p}}\right)
    \end{equation}
    queries to the state preparation oracle of $\ket{b}$. 
    Here $1/\sqrt{p} = \|A^{-1}\|/\|A^{-1}\ket{b}\|$, and we assume an estimate of $\|A^{-1}\ket{b}\|$ up to constant multiplicative error is known a priori. 
    \end{enumerate}
\end{theorem}

\section{Conclusion and outlook}\label{sec:conclusion}

In this work, we developed a transducer-based LCU framework that replaces the multiplicative worst-case cost dependence of standard LCU by a weighted-average dependence under a shared primitive-access assumption. We also developed resolvent-based tools for composing and implementing transducers, and applied the resulting framework to sparse matrices, obtaining near-optimal query complexities for sparse Hamiltonian simulation and quantum linear systems up to polylogarithmic factors.

A natural first direction is to extend the ideas developed here beyond LCU to QSVT. QSVT is a considerably more general framework for quantum matrix algorithms, and it would be interesting to understand whether transducer techniques can similarly exploit nonuniform costs within QSVT constructions and lead to refined complexity bounds. Such an extension could potentially provide a systematic route to further applications.

A second direction is to further develop the transducer framework itself. Our analysis identifies the resolvent norm $K(S)$ as a useful complexity measure alongside the catalyst complexity. It would be interesting to clarify the relation between these two measures, establish sharper bounds for the resolvent norm under transducer composition, and determine whether the $\mathcal{O}(K(S)\log(1/\epsilon))$ complexity of our finite-reuse implementation can be further improved.

\section*{Acknowledgements}

DA, DD, and CZ acknowledge funding from Quantum Science and Technology - National Science and Technology Major Project via Project 2024ZD0301900, and the support by The Fundamental Research Funds for the Central Universities, Peking University. 
CS was supported by the National Key Research Project of China under Grant No. 2025YFA1017200.
YZ was supported by the European Research Council (ERC) via the Starting Grant QuanThermal (101222179).

\section*{AI use statement}

This work involved extensive interaction between the authors and a large language model (LLM). 
The authors initially aimed to develop a quantum algorithm for sparse linear systems with optimal dependence on all relevant parameters. 
During this stage, the LLM was used for brainstorming and suggested using the original transducer framework~\cite{BelovsJefferyYolcu2024}, which could achieve the desired sparsity dependence but not the target precision dependence. 
After the appearance of~\cite{chen2026optimal}, the authors recognized that an improved implementation of transducers could yield better precision dependence and discussed this idea with LLM. 
The LLM then assisted in developing an early version of sparse linear system solver with near-optimal dependence on all parameters.

While refining that early version, the authors realized that the techniques could be extended to the broader task of LCU. The authors consequently shifted the motivation and focus of the work to developing an improved LCU primitive through the transducer framework, with sparse matrix functions serving as applications. 
The authors formulated the central idea and key steps of the transducer-based LCU algorithm, while the LLM assisted with the analysis and the development of technical results. 
Building on these discussions, the authors further developed and refined the analysis and proofs.

The authors wrote the initial manuscript, and used the LLM to refine the presentation of selected sections. 
The authors take full responsibility for all content presented in this work.

\bibliographystyle{unsrt}
\bibliography{refs}

\appendix

\section{Polynomial approximation of the inverse Cayley transformation and amplification}\label{sec:polynomial_inv_Cayley_amplification}

We use the following result on polynomial approximations of analytic functions.
\begin{lemma}[Corollary 66 of~\cite{GilyenSuLowEtAl2019}]\label{lem:GSLW_taylor}
    Let $x_0\in[-1,1]$, $r\in(0,2]$, $\eta\in(0,r]$, and let $f\colon[x_0-r-\eta,\,x_0+r+\eta]\to\mathbb{C}$ be such that $f(x_0+x)=\sum_{\ell=0}^{\infty}a_\ell x^\ell$ for all $x\in[-r-\eta,\,r+\eta]$. Suppose $B>0$ is such that $\sum_{\ell=0}^{\infty}(r+\eta)^\ell|a_\ell|\leq B$, and let $\epsilon\in(0,\frac{1}{2B}]$. Then there is an efficiently computable polynomial $\widetilde P\in\mathbb{C}[x]$ of degree $\mathcal{O}\big(\frac{1}{\eta}\log\frac{B}{\epsilon}\big)$ such that
    \be
        \big\|f-\widetilde P\big\|_{[x_0-r,\,x_0+r]}\leq\epsilon
        \quad\text{and}\quad
        \big\|\widetilde P\big\|_{[-1,1]}\leq\epsilon+B.
    \ee
\end{lemma}

\begin{lemma}\label{lem:polynomial_inv_Cayley_amplification}
    For any $0 < \rho \leq 1/4$ and $0 <\delta<1/2$, there exists a real odd polynomial $P(x)$ such that 
    \begin{align}
        |P(x)| &\leq \tfrac{1}{2}, \quad \forall x \in [-1,1], \\
        \left| P\left( \frac{2b}{1+b^2} \right) - \frac{b}{16\rho} \right| & \leq \delta, \quad \forall b \in [-\rho,\rho], \\
        \operatorname{deg}(P(x)) &= \mathcal{O}\left( \frac{1}{\rho} \log\left( \frac{1}{\delta} \right) \right). 
    \end{align}
\end{lemma}

\begin{proof}
    Let $g(x) \coloneqq \frac{x}{1+\sqrt{1-x^2}}$. For $|b|\leq 1$, a direct computation gives $g\big(\frac{2b}{1+b^2}\big) = b$, and $\big|\frac{2b}{1+b^2}\big|\leq 2\rho$ whenever $|b|\leq\rho$. Hence it suffices to approximate $f(x)\coloneqq g(x)/(16\rho)$ on $[-2\rho,2\rho]$.
    Since $g(x) = (1-\sqrt{1-x^2})/x = \sum_{k\geq 1} c_k x^{2k-1}$ with $c_k\geq 0$ for $|x|\leq 1$, the coefficients $a_\ell$ of $f$ satisfy
    \be
        \sum_\ell (4\rho)^\ell |a_\ell| = \frac{g(4\rho)}{16\rho} \leq \frac14.
    \ee
    Let $\epsilon_0 \coloneqq \min\{\delta, 1/16\}$. Applying \cref{lem:GSLW_taylor} with $x_0=0$, $r=\eta=2\rho$, $B=1/4$ and $\epsilon=\epsilon_0$ yields a polynomial $\widetilde P$ of degree $\mathcal{O}(\rho^{-1}\log(1/\delta))$ such that
    \be
        \|f-\widetilde P\|_{[-2\rho,2\rho]} \leq \epsilon_0
        \quad\text{and}\quad
        \|\widetilde P\|_{[-1,1]} \leq \epsilon_0 + \tfrac14 \leq \tfrac{5}{16}.
    \ee
    Since $f$ is real and odd, the odd part of the real part of $\widetilde P$, i.e.\ $P(x)\coloneqq \frac12\big(\operatorname{Re}\widetilde P(x) - \operatorname{Re}\widetilde P(-x)\big)$, satisfies the same bounds. Specifically, we have
    \be
    \|P\|_{[-1,1]} < \tfrac{1}{2}\quad\text{and}\quad
    \|f-P\|_{[-2\rho,2\rho]}\leq \epsilon_0\leq \delta,
    \ee
    which completes the proof.
\end{proof}

\section{A simple explanation of the construction in Lemma \ref{lem:composition_transducer}}
\label{app:B}

The construction in Lemma \ref{lem:composition_transducer} admits a natural interpretation in terms of a random walk on a directed cycle. We illustrate this for $L=3$. Consider
\be
S=
\begin{pmatrix}
0&0&V_2\\
V_0&0&0\\
0&V_1&0
\end{pmatrix}.
\ee
This operator, inspired by \eqref{eqn:proof_unitary_composition_def_S}, can be viewed as a one-step walk operator on a directed cycle with vertices $0,1,2$, where the transition from vertex $j$ is accompanied by the unitary $V_j$. Indeed,
\be
\ket{0}\otimes\ket{x}
\xrightarrow{S}
\ket{1}\otimes V_0\ket{x}
\xrightarrow{S}
\ket{2}\otimes V_1V_0\ket{x}
\xrightarrow{S}
\ket{0}\otimes V_2V_1V_0\ket{x}.
\ee
Equivalently, if
\be
S
\begin{pmatrix}
x\\
y_1\\
y_2
\end{pmatrix}
=
\begin{pmatrix}
Vx\\
y_1\\
y_2
\end{pmatrix},
\qquad
V=V_2V_1V_0,
\ee
then  
$
y_1=V_0x,
y_2=V_1V_0x.
$
The same picture extends to the composition of transducers. Let
\be
S_i=
\begin{pmatrix}
D_i&E_i\\
B_i&T_i
\end{pmatrix},
\ee
and consider the direct sum $S_0\oplus S_1\oplus S_2$. Grouping the public variables first and the private variables second, we have
\be
\begin{pmatrix}
x_0\\
x_1\\
x_2\\
y_0\\
y_1\\
y_2
\end{pmatrix}
\xrightarrow{\,S_0\oplus S_1\oplus S_2\,}
\begin{pmatrix}
D_0x_0+E_0y_0\\
D_1x_1+E_1y_1\\
D_2x_2+E_2y_2\\
B_0x_0+T_0y_0\\
B_1x_1+T_1y_1\\
B_2x_2+T_2y_2
\end{pmatrix}.
\ee
To reproduce the directed cycle $0\to1\to2\to0$, we cyclically shift only the public outputs while leaving the private outputs at their original vertices. This gives
\be
\begin{pmatrix}
D_2x_2+E_2y_2\\
D_0x_0+E_0y_0\\
D_1x_1+E_1y_1\\
B_0x_0+T_0y_0\\
B_1x_1+T_1y_1\\
B_2x_2+T_2y_2
\end{pmatrix}
=
\begin{pmatrix}
0&0&D_2&0&0&E_2\\
D_0&0&0&E_0&0&0\\
0&D_1&0&0&E_1&0\\
B_0&0&0&T_0&0&0\\
0&B_1&0&0&T_1&0\\
0&0&B_2&0&0&T_2
\end{pmatrix}
\begin{pmatrix}
x_0\\
x_1\\
x_2\\
y_0\\
y_1\\
y_2
\end{pmatrix}.
\ee
This is exactly the construction in
\eqref{eqn:proof_transducer_composition_def_S} for $L=3$.

This walk interpretation also helps explain the bound on $K(S)$. Consider
\be
(I-T)
\begin{pmatrix}
h_1\\
h_2\\
v_0\\
v_1\\
v_2
\end{pmatrix}
=
\begin{pmatrix}
a_1\\
a_2\\
b_0\\
b_1\\
b_2
\end{pmatrix}.
\ee
Writing
\be
R_j=(I-T_j)^{-1},
\qquad
\Gamma_j=R_jB_j,
\ee
the local private equations give
\be
\begin{pmatrix}
v_0\\
v_1\\
v_2
\end{pmatrix}
=
\begin{pmatrix}
R_0&0&0\\
0&R_1&0\\
0&0&R_2
\end{pmatrix}
\begin{pmatrix}
b_0\\
b_1\\
b_2
\end{pmatrix}
+
\begin{pmatrix}
0\\
\Gamma_1h_1\\
\Gamma_2h_2
\end{pmatrix}.
\ee
which is exactly \eqref{eqn:proof_transducer_composition_rec_eq3}.
The first term is purely local and is therefore controlled by
\be
\max_j\|R_j\|=\max_jK(S_j).
\ee
The second term is generated by the intermediate clock variables $h_1,h_2$, which collect residuals propagated along the directed walk. 
To be a little more exact,
the second term is generated by the intermediate clock variables $h_1,h_2$. These variables themselves contain residuals propagated from earlier vertices. Indeed,
\be
h_1
=
a_1+V_0\Gamma_0^\dagger b_0,
\quad
h_2
=
V_1h_1+a_2+V_1\Gamma_1^\dagger b_1.
\ee
Thus a local residual $b_j$ first enters the clock propagation through the map
$
\Gamma_j^\dagger,
$
is transported along the directed walk by the unitaries $V_j$, and is finally lifted back into the local private variables through the maps $\Gamma_i$. Since the $V_j$'s are unitary, the propagation itself does not increase the norm, while the entrance and exit maps are controlled by
$
\|\Gamma_j\|=\sqrt{w(S_j)}.
$
More precisely, the map from a pair of residuals $(a_j,b_{j-1})$ to the corresponding clock source has norm at most
\be
\left\|
\begin{pmatrix}
I & V_{j-1}\Gamma_{j-1}^\dagger
\end{pmatrix}
\right\|
\leq
\sqrt{1+\max_j w(S_j)},
\ee
whereas reconstructing the public/private pair
$
h_j\longmapsto
\begin{pmatrix}
h_j\\
\Gamma_j h_j
\end{pmatrix}
$
costs another factor of at most
$
\sqrt{1+\max_j w(S_j)}.
$
The residuals can propagate through at most $2=L-1$ intermediate clock positions. Consequently, the propagation contribution is bounded by
\be
2\bigl(1+\max_j w(S_j)\bigr).
\ee

Thus the bound on $K(S)$ in Lemma \ref{lem:composition_transducer} naturally separates into a local elimination contribution, governed by $\max_j K(S_j)$, and a propagation contribution arising from the accumulation of residuals along the clock path. The former reflects the worst local sensitivity of the individual transducers. Regarding the latter, the factor $L-1$ comes from propagation through the intermediate clock variables, whereas $w(S_j)=\|\Gamma_j\|^2$ controls the coupling between the clock variables and the corresponding local private spaces. In this sense, the two terms represent two distinct mechanisms: local Schur elimination and global propagation along the clock. This decomposition is essentially tight at the level of scaling: in general, neither the dependence on $\max_j K(S_j)$ nor the linear accumulation in $L$ and $1+\max_j w(S_j)$ can be removed.

\paragraph{Connections to products of block-encodings.}

The clock construction in Lemma~\ref{lem:composition_transducer}
provides a query-efficient method for composing multiple transducers.
Here we describe an alternative sequential construction, which is more closely
related to the usual composition of block encodings \cite{GilyenSuLowEtAl2019}. Although this construction
does not necessarily improve the worst-case bound on \(K(S)\) compared with the
clock construction, its Schur-complement structure reveals a different
multiplicative behavior of the resolvent norm.

We start from the case \(L=2\).
Let
\be
S_j=
\begin{pmatrix}
D_j&E_j\\
B_j&T_j
\end{pmatrix},
\qquad j=0,1,
\ee
be transducers of \(V_j\) with catalyst maps
\(\Gamma_j\). We embed them into the common space
$\mathcal H_{ pub}\oplus
\mathcal H_0\oplus
\mathcal H_1$
and define
\be
\widehat S_0=
\begin{pmatrix}
D_0&E_0&0\\
B_0&T_0&0\\
0&0&I
\end{pmatrix},
\qquad
\widehat S_1=
\begin{pmatrix}
D_1&0&E_1\\
0&I&0\\
B_1&0&T_1
\end{pmatrix}.
\ee
The sequential composition 
\be
\label{eq1:caseL=2}
S_{\rm seq}=\widehat S_1 \widehat S_0
\ee
is a unitary operator. One may realize that $S_{\rm seq}$ is indeed the standard way of constructing block-encodings of the product $D_1D_0$.
Define the catalyst map
\be
\Gamma_{\rm seq}x
=
\begin{pmatrix}
\Gamma_0x\\
\Gamma_1V_0x
\end{pmatrix}.
\ee
Using $D_j+E_j\Gamma_j=V_j,
B_j+T_j\Gamma_j=\Gamma_j,$
one directly verifies that
\be
S_{\rm seq}
\begin{pmatrix}
x\\
\Gamma_{\rm seq}x
\end{pmatrix}
=
\begin{pmatrix}
V_1V_0x\\
\Gamma_{\rm seq}x
\end{pmatrix}.
\ee
Hence \(S_{\rm seq}\) is a transducer of the product
$V=V_1V_0 .$
The private block of \(S_{\rm seq}\) is
\be
T_{\rm seq}
=
\begin{pmatrix}
T_0&0\\
B_1E_0&T_1
\end{pmatrix}.
\ee
Therefore,
\be
(I-T_{\rm seq})^{-1}
=
\begin{pmatrix}
R_0&0\\
R_1B_1E_0R_0&R_1
\end{pmatrix},
\ee
where $R_j=(I-T_j)^{-1}.$
Using the transducer identities $R_jB_j=\Gamma_j$ and $E_jR_j=V_j\Gamma_j^\dagger,$
the off-diagonal block becomes $R_1B_1E_0R_0
=
\Gamma_1V_0\Gamma_0^\dagger$.
Hence $\|R_1B_1E_0R_0\|
\leq
\sqrt{w(S_1)w(S_0)} .$
Consequently,
\be
K(S_{\rm seq})
=
\|(I-T_{\rm seq})^{-1}\|
\leq
\max\{K(S_0),K(S_1)\}
+
\sqrt{w(S_0)w(S_1)} .
\ee
Compared with the clock construction, the propagation term depends on the
geometric mean of the catalyst complexities rather than the worst-case
catalyst weight.

The above phenomenon also appears for 
\(L\geq 3\). We briefly discuss the case $L=3$. Now we introduce three private spaces and define
\be
\widehat S_0=
\begin{pmatrix}
D_0&E_0&0&0\\
B_0&T_0&0&0\\
0&0&I&0\\
0&0&0&I
\end{pmatrix},
\quad 
\widehat S_1=
\begin{pmatrix}
D_1&0&E_1&0\\
0&I&0&0\\
B_1&0&T_1&0\\
0&0&0&I
\end{pmatrix},
\quad
\widehat S_2=
\begin{pmatrix}
D_2&0&0&E_2\\
0&I&0&0\\
0&0&I&0\\
B_2&0&0&T_2
\end{pmatrix}.
\ee
Set
\be
\label{eq1:caseL=3}
S_{\rm seq}
=
\widehat S_2\widehat S_1\widehat S_0 ,
\ee
which corresponds to the block-encoding of $D_2D_1D_0$.
The corresponding catalyst map is
\be
\Gamma_{\rm seq}x
=
\begin{pmatrix}
\Gamma_0x\\
\Gamma_1V_0x\\
\Gamma_2V_1V_0x
\end{pmatrix}.
\ee
Again, applying the defining relations of each transducer gives
\be
S_{\rm seq}
\begin{pmatrix}
x\\
\Gamma_{\rm seq}x
\end{pmatrix}
=
\begin{pmatrix}
V_2V_1V_0x\\
\Gamma_{\rm seq}x
\end{pmatrix}.
\ee
Thus the sequential construction realizes the desired product action. A similar calculation shows that the following bound
\be
K(S_{\mathrm{seq}})
\lesssim
\max K(S_j)
+
\sqrt{w(S_0) w(S_1)}
+
\sqrt{w(S_1)w(S_2)}
+
\sqrt{w(S_0)w(S_2)} w(S_1) .
\ee

In summary, the sequential construction has a close connection with the standard product
of block-encodings. This
connection provides a natural Schur-complement interpretation of products of
block encodings and reveals a multiplicative coupling structure in the
resolvent norm \(K(S)\). However, the resulting \(K(S)\) bound can be worse due
to the appearance of higher-order products of catalyst weights.

In contrast, the clock construction in Lemma~\ref{lem:composition_transducer}
is motivated by the core idea of LCU: all components are coherently selected
within an enlarged Hilbert space, while the intermediate computation is stored
in the private space. Although this construction may lose some of the finer
multiplicative structure of the sequential composition, it has important
algorithmic advantages. Under the shared primitive-access assumption \ref{assump:U}, the
select oracle of the individual transducers can be implemented efficiently,
so the whole composed transducer can be realized with a single coherent query
to the underlying primitive oracles. Consequently, the action can be
implemented with cost $\mathcal{O}(K(S)\log(1/\epsilon))$ described in Theorem \ref{thm:finite-reuse-high-order}.
For the sequential construction, implementing the product requires applying
each \(S_j\) individually, resulting in an additional factor depending on the
composition length \(L\).

\end{document}